\documentclass[%
 reprint,
 amsmath,amssymb,
 aps,
]{revtex4-2}

\usepackage{graphicx}
\usepackage{dcolumn}
\usepackage{bm}
\usepackage{braket}
\usepackage{amsthm}
\usepackage{amssymb}
\newtheorem{theorem}{Theorem}
\newtheorem*{note}{Note}
\newtheorem*{lemma}{Lemma}

\begin{document}

\preprint{APS/123-QED}

\title{Foundation for the $\Delta SCF$ Approach in Density Functional Theory}

\author{Weitao Yang}
\affiliation{Department of Chemistry and Department of Physics, Duke University, Durham, North Carolina 27708}
 \email{weitao.yang@duke.edu}
\author{Paul W. Ayers}
\affiliation{Department of Chemistry and Chemical Biology, McMaster University, Hamilton, Ontario L8S 4M1}
 \email{ayers@mcmaster.ca}

\date{\today}

\begin{abstract}
We extend ground-state density-functional theory to excited states
and provide the theoretical formulation for the widely
used $\Delta SCF$ method for calculating excited-state energies and densities. As the electron density alone is insufficient to characterize
excited states, we formulate excited-state theory using the defining variables
of a noninteracting reference system, namely (1) the excitation quantum number
$n_{s}$ and the potential $w_{s}(\mathbf{r})$ (excited-state potential-functional theory, $n$PFT), (2) the noninteracting wavefunction $\Phi$
($\Phi$-functional theory, $\Phi$FT), or (3) the noninteracting one-electron reduced
density matrix $\gamma_{s}(\mathbf{r},\mathbf{r}')$ (density-matrix-functional theory, $\gamma_{s}$FT). We show the equivalence of
these three sets of variables and their corresponding energy functionals. Importantly, the ground and excited-state exchange-correlation energy use the \textit{same} universal functional, regardless of whether $\left(n_{s},w_{s}(\boldsymbol{r})\right)$, $\Phi$, or $\gamma_{s}(\mathbf{r},\mathbf{r}')$ is selected as the fundamental descriptor of the system. We derive the excited-state (generalized) Kohn-Sham equations. The minimum of all three functionals is the ground-state energy and, for ground states, they are all equivalent to the Hohenberg-Kohn-Sham method. The other stationary points of the functionals provide  the excited-state energies and electron densities, establishing the foundation for the $\Delta SCF$ method.
\end{abstract}

\maketitle

The mathematical framework of density-functional theory (DFT), in its original formulations, is based on the minimum-energy variational principle and is thus restricted to ground 
states.\cite{HohenbergKohn1964,kohnSelfconsistentEquationsIncluding1965,Levy1979constrainedsearch,valoneOnetoOneMappingOneParticle1980,Lieb1983dftIJQC,Lieb1985,ayersAxiomaticFormulationsHohenbergKohn2006,TealeHelgaker2022DFTexchange} However, the importance of photochemical and nonadiabatic dynamics provides impetus for a treatment of excited states with DFT. 
Early successes allowed ground-state DFT to be extended to the lowest-energy excited state of a given symmetry.\cite{vonBarthHedin1972SDFT,GunnarssonLundqvist1976, PerdewZunger1981spinDFT, Gorling1993SymmetryDFT, Theophilou1997ensembleDFT, Theophilou1998essymmetries}, but this requires symmetry-specific functionals and only provides access to a tiny fraction of excited states. 

The theoretical formulation of ground-state DFT cannot be directly extended to excited states because there is no excited-state extension of the Hohenberg-Kohn theorem.\cite{GaudoinBurke2004esHKtheorem, GaudoinBurke2005esHKtheoremerrata, SamalHarbola2006esHK, SamalHarbolaHolas2006esHK} Specifically, it is possible for the $m^{\text{th}}$ excited state density of one system to be the $n^{\text{th}}$ excited state density of a different system. Therefore one needs more than just the excited-state electron density to fully specify the state of the system. This additional information can come from:
(1) A real number, or pair of integers, specifying the excitation level, \cite{Gorling1993SymmetryDFT, Gorling1996esDFT, Gorling1999esDFTgenAC, Gorling2000esDFTPRL,RNAyersLevy2009excited} (2) The ground-state density and the excitation level,\cite{Gorling1996esDFT, LevyNagy1999esDFTKoopmans, LevyNagy1999esDFTPRL, NagyLevy2001esDFTconstrainedsearch, RNAyersLevy2009excited, NagyLevyAyers2009esDFTchapter,HarbolaSamal2009esHK, HarbolaSamal2009esHKchapter}
 (3) ensemble weights,\cite{Theophilou1979ensembleDFT,GrossOliveiraKohn1988ensembleDFT,OliveiraGrossKohn1988ensembleDFT,OliveiraGrossKohn1990ensembleDFT,TheophilouGidopoulos1995ensembleDFT,Gidopoulos2002ensembleDFT,GidopoulosGross2002ensembleDFT,FranckFromager2014ensembleDFT,Filatov2015ensembleDFT,CernaticSenjeanFromager2021ensembleDFT,DeurFromager2019ensembleDFT,fromagerIndividualCorrelationsEnsemble2020,GouldKronik2021ensembleDFTgenKS,GouldKooiGoriGiorgiPittalis2023ensembleDFT}
 or (4) a matrix of densities.\cite{KleinDreizler2002esDFT, GaoGrofe2016esMatrixDFT, LuGao2022matrixDFTesJCTC, LuGao2022matrixDFTes}

There are also approaches based on ground-state \textit{response} properties. 
For example, the  dynamic linear response of the ground state diverges when the frequency is matched to an excitation energy. This motivates time-dependent DFT (TDDFT),\cite{RungeGross1984tddft,GrossKohn1985tddft, CasidaChong1995TDDFTreview, Casida2009reviewTDDFT} as well as closely related approaches based on the random phase approximation (including particle-hole and particle-particle/hole-hole flavors), equations of motion, and propagator theory.\cite{Yangvanaggelen2013ppRPAes, YangPenLuYang2014ppRPAes,ZhangPengYang2015ppRPAes,YangPengDavidson2015ppRPAes,YangDavidsonYang2016ppRPAes, MeiYang2019qedft} 
Establishing the exactness of such approaches is significantly harder than the analogous Levy-Lieb formulations of ground-state DFT.\cite{vanLeeuwen1999tddftrigorous,vanLeeuwen2001tddftrigorous, RuggenthalerVanleeuwen2015tddftrigorous}

An attractive \textit{practical} approach is the $\Delta$SCF method, which was first proposed in the context of what we now call DFT by Slater.\cite{Slater1972AdvQuantumChem,SlaterWood1970SlaterTM} In $\Delta$SCF, one constructs the excited-state density by choosing a non-aufbau population of the Kohn-Sham (KS) orbitals and converges the associated self-consistent field (SCF) calculation; the energy is evaluated using the \textit{ground-state} density functional even though one is targeting an excited state.\cite{ZieglerRaukBaerends1977deltaSCF,JonesGunnarsson1989dftreview}  Results are often excellent, comparable to the accuracy of TDDFT for well-behaved excited states but superior for excitations with charge-transfer or multiple-excitation character.\cite{TrigueroPetterssonAgren1999deltaSCF,ChengWuVanVoorhis2008DeltaSCF, GilbertBesleyGill2008MOM, SeiduZiegler2015constrictedDFT,LiuBaoYi2017deltaSCF,BarcaGilbertGill2018deltaSCFes, BarcaGilbertGill2018iMOM,  HaitHeadgordon2020deltaSCF, KumarLuber2022deltaSCF, GrimmeMewes2021deltaSCF} This begs the question: why can one apply a density functional that was developed for ground states to excited states? 

To formulate the $\Delta$SCF approach as a density-functional method, one needs to define an energy functional that is stationary for excited excited states.\cite{Gorling1999esDFTgenAC} For the restricted set of excited-state electron densities that are not the ground-state density for any system, the Levy constrained search functional is stationary.\cite{PerdewLevy1985esDFT} To extend this result to arbitrary excited states, G\"{o}rling replaced the minimization of the Levy constrained search functional with a stationary principle and labeled the stationary states with a real number $\nu$.\cite{Gorling1999esDFTgenAC}, leading to a procedure where the functionals depend on the $\nu$ that specifies the system and state of interest. Indeed, the theoretical foundation of the $\Delta$SCF method as it is used in practice has not been established.\cite{VandaeleLuber2002deltaSCFreview,Herbert2023esDFTchapter}

Here, we show that the $\Delta$SCF method is exact and reveal the universal ground-and-excited-state energy functional. Instead of the density, our approach uses the potential as the basic variable and uses the fact that $\Delta$SCF theory only requires that excited-state electron densities, which are inherently interacting  \textit{v}-representable, also be noninteracting \textit{v}-representable. We use the fundamental excited-state potential functionals to build practical $\Delta$SCF methods based on the stationary-state wavefunctions and density matrices of the noninteracting system.

Let $\Psi_{n,w}^\lambda$ and $\rho_{n,w}^\lambda(\mathbf{r})$ denote the wavefunction and density of the $n$th eigenstate of a $N$-electron Hamiltonian with a local potential, 
\begin{equation}\label{eq:H_w}
\hat{H}_{w}^\lambda=\hat{T}+\lambda V_{ee}+\sum_{i=1}^{N}w(\mathbf{r}_{i}).
\end{equation}
Extending the concept of ground-state \textit{v}-representability, the wavefunctions, density matrices, and densities of eigenstates from Hamiltonians of this form are said to be excited-state \textit{v}-representable. Extending the ground-state potential functional theory for a physical system with external potential $v(\mathbf{r})$,\cite{YangAyersWu2004potentialFT} define the excited-state potential functional:
\begin{align}\label{eq:E1}
E_{v}^{\lambda}[n,w] & = \braket{\Psi_{n,w}^\lambda|\hat{H}_{v}^\lambda|\Psi_{n,w}^\lambda} \\
&= \braket{ \Psi_{n,w}^\lambda|\hat{H}_{w}^\lambda|\Psi_{n,w}^\lambda} +\int d\mathbf{r}\left(v(\mathbf{r})-w (\mathbf{r})\right)\rho_{n,w}^\lambda(\mathbf{r}).\nonumber
\end{align}
\begin{theorem}
$E_{v}^{\lambda}[n,w]$ is stationary
with respect to variations in the trial potential, $w(\mathbf{r})$, if 
$w(\mathbf{r})$ and the physical potential, $v(\mathbf{r})$, differ by at most a constant. 
\end{theorem}
\begin{proof} 
Consider a variation $\delta w(\mathbf{r})$.
For a nondegenerate eigenstate, $\delta\left\langle \Psi_{n,w}^\lambda\right|\hat{H}_{w}^\lambda\left|\Psi_{n,w}^\lambda \right\rangle =\int d\mathbf{r}\delta w(\mathbf{r})\rho_{n,w}^\lambda(\mathbf{r})$
,
\begin{align}
\delta E_{v}^{\lambda}[n,w]  =&\delta\left\langle \Psi_{n,w}^\lambda\right|\hat{H}_{w}^\lambda\left|\Psi_{n,w}^\lambda\right\rangle -\int d\mathbf{r}\delta w(\mathbf{r})\rho_{n,w}^\lambda(\mathbf{r})\notag
\\&+\iint d\mathbf{r}'d\mathbf{r}\left(v(\mathbf{r}')-w(\mathbf{r}')\right)\frac{\delta\rho_{n,w}^\lambda(\mathbf{r}')}{\delta w(\mathbf{r})}\delta w(\mathbf{r}) \notag
\\ =&\iint d\mathbf{r}'d\mathbf{r}\left(v(\mathbf{r})-
 w(\mathbf{r})\right)\frac{\delta\rho_{n,w}^\lambda(\mathbf{r})}{\delta w(\mathbf{r}')}\delta w(\mathbf{r}') \label{eq:th1}
\end{align}
$\delta E_{N,v}^\lambda[n,w(\mathbf{r})]=0$ when $v(\mathbf{r})-w(\mathbf{r})$ is a constant. 
\end{proof}
\begin{note}See the supplemental material for the theorem's converse and its extension to degenerate states. This result, and all subsequent analysis, is easily extended to spin-resolved (unrestricted) calculations: just replace the spatial variable $(\mathbf{r})$ with the spin variable $(\mathbf{x})=(\mathbf{r},\sigma)$.\end{note}

Theorem 1 allows us to shift variables from the many-electron wavefunction to the potential: the $n^{\text{th}}$ excited state can be determined by making the energy functional $E_{v}^{1}[n,w]$, stationary with respect to the potential $w$. The potential and $n$ then suffice to determine the excited-state wavefunction, $\Psi_{n,v}^\lambda$, (by solving the Schr{\"o}dinger equation). 

Conversely, because Hamiltonians, $\hat{H}_w^\lambda$, with different potentials have different eigenfunctions and different eigenfunctions of the same Hamiltonian are not only different, but orthogonal,
\begin{lemma}
There is a mapping from the space of excited-state v-representable wavefunctions to their corresponding potentials and excitation levels. 
\end{lemma}
This lemma is proved in the supplementary material.

We now address practical KS-DFT calculations. The \emph{ground-state Kohn-Sham assumption} is that for the ground state of any physical system, there exists a noninteracting system with the same ground-state density, $\rho_{0,w}^1(\mathbf{r}) = \rho_{0,w_s}^0(\mathbf{r})$. (I.e., every ground-state $v$-representable density is assumed to be ground-state noninteracting $v$-representable.) The Hohenberg-Kohn theorem then guarantees that the ground-state density uniquely determines the (non)interacting potential, $w_s(\mathbf{r})$, and the number of electrons, ergo the noninteracting Hamiltonian
\begin{equation}
\hat{H}_{s}=\hat{T}+\sum_{i}^{N}w_{s}(\mathbf{r}_{i})=\hat{H}_{w_s}^0,\label{eq:H_s}
\end{equation}
and its eigenfunctions, which can be chosen to be Slater determinants, $\Phi$. 

We first recognize that $\Delta$SCF calculations rely upon the \emph{excited-state Kohn-Sham assumption}, namely that for any bound state of any physical system, there exists a bound state of a noninteracting system with the same electron density, $\rho_{n,w}^1(\mathbf{r}) = \rho_{n_s,w_s}^0(\mathbf{r})$. (I.e., each excited-state $v$-representable density is also excited-state noninteracting $v$-representable.) 
For the same electron density, there may be multiple noninteracting systems with different excitation quantum numbers $n_{s}$ and local potentials $w_{s}(\mathbf{r})$. 
Starting from any of these noninteracting systems, the many-electron wavefunction $\Psi_{n,w}$ can be approached from the noninteracting $\Phi_{n_{s}w_{s}}$ through an adiabatic
connecting Hamiltonian, $\hat{H}_{w^\lambda}^\lambda$ (cf. Eq. (\ref{eq:H_w})), where $w^{\lambda}(\mathbf{r})$ is any parameterized path of local potentials that satisfies the boundary conditions: $w^{0}(\mathbf{r})=w_{s}(\mathbf{r})$ and $w^1(\mathbf{r})=w(\mathbf{r})$.\cite{LangrethPerdew1977adiabaticC,harrisAdiabaticconnectionApproachKohnSham1984a,yangGeneralizedAdiabaticConnection1998a} (For example, one may choose $w_\lambda(\mathbf{r}) = (1-\lambda)w_s(\mathbf{r}) + \lambda w(\mathbf{r})$.\cite{harrisSurfaceEnergyBounded1974a}) The adiabatic connection maps the excitation quantum number $n$ and the local potential $w(\mathbf{r})$, or equivalently, $\Psi$, of a bound state of an interacting system to the corresponding quantities of a computationally convenient noninteracting system, $n_{s}$ and $w_{s}(\mathbf{r})$, or equivalently, $\Phi_{n_{s}w_{s}}$.

Our excited state KS assumption is parallel to the KS assumption for
the ground states, with one critical difference: we do not
assume the one-to-one mapping between excited-state densities and local potentials. It is not even necessary to choose the noninteracting reference so that its excitation number, $n_s$, matches that of the physical system: to construct the adiabatic connection we only need assume that there exists \textit{some} eigenstate of \textit{some} noninteracting reference system with the same density as the targeted eigenstate of the physical system. 

We now use the excited-state KS assumption to construct the energy functionals and variational principles for $\Delta$SCF calculations. Using the excited state KS mapping from $\{n_s,w_s(\mathbf{r})\}$
to $\{n,w(\mathbf{r})\}$, we can reexpress the excited-state potential functional, Eq. (\ref{eq:E1}), in terms of the noninteracting reference system,
\begin{align}\label{eq:E0}
E_{v}[n_{s},w_{s}]  =E_{v}^{1}\left[n,w\right] =\Braket{\Psi_{n,w}^1|\hat{H}_{v}^1 |\Psi_{n,w}^1}. 
\end{align}
As a consequence of the Lemma, an excited-state noninteracting $v$-representable wavefunction, $\Phi$,
uniquely determines $n_{s}$ and $w_{s}$; we can thus define a functional of the KS wavefunction, 
\begin{align}\label{eq:PhiFT}
E_{v}\left[\Phi\right] & =E_{v}[n_{s},w_{s}].
\end{align}
Similarly, as shown in the supplementary material, an excited-state one-electron density matrix, $\gamma_s(\mathbf{r},\mathbf{r}')$, determines $n_{s}$ and $w_{s}$. We thus define a (noninteracting) density-matrix functional, 
\begin{equation}\label{eq:RhosFT}
E_{v}\left[\gamma_{s}(\mathbf{r},\mathbf{r}')\right]=E_{v}[n_{s},w_{s}].
\end{equation}
We call these three equivalent approaches the excited-state potential-functional
theory ($n$PFT), the $\Phi$-functional theory ($\Phi$FT), and the
$\gamma_{s}$-functional theory ($\gamma_{s}$FT). These functionals, together with the following variational principle, provide the foundation for the $\Delta$SCF method. 

\begin{theorem} $E_{v}[n_{s},w_{s}]$ is stationary
with respect to variations in the trial potential $w_{s}(\mathbf{r})$,
when $w(\mathbf{r})$, the map of $w_{s}(\mathbf{r})$
to the potential at full interaction strength, equals the
physical potential $v(\mathbf{r})$ (up to a constant).
\end{theorem}
\begin{proof}   
Consider a variation $\delta w_{s}(\mathbf{r})$.
From Eqs. (\ref{eq:E1}) and (\ref{eq:E0}),
\begin{align}
\frac{\delta E_{v}}{\delta w_{s}(\mathbf{r})}  =&\int d\mathbf{r}'\biggl[\frac{\delta\left\langle \Psi_{n,w}\right|H_{w}\left|\Psi_{n,w}\right\rangle }{\delta w(\mathbf{r}')}-\rho_{n,w}(\mathbf{r})\nonumber\\
+&\int d\mathbf{r}''\left(v(\mathbf{r}'')-w(\mathbf{r}'')\right)\frac{\delta\rho_{n,w}(\mathbf{r}'')}{\delta w(\mathbf{r}')}\biggr]\frac{\delta w(\mathbf{r}')}{\delta w_{s}(\mathbf{r})}\nonumber\\
 =&\int d\mathbf{r}''\left(v(\mathbf{r}'')-w(\mathbf{r}'')\right)\frac{\delta\rho_{n,w}(\mathbf{r}'')}{\delta w_{s}(\mathbf{r})}
\end{align}
if $v(\mathbf{r})-w(\mathbf{r})$ is a constant, $\delta E_{v}[n_s,w_{s}(\mathbf{r})]=0.$
\end{proof}
\begin{note} See the supplemental material for the theorem's converse and its extension to degenerate states.
\end{note}

The variational principle for $E_{v}[n_{s},w_{s}]$ with respect
to $w_{s}$ implies corresponding variational principles for $E_{v}\left[\Phi\right]$ and $E_{v}\left[\gamma_{s}\right]$. E.g., from the definition, (\ref{eq:PhiFT}),
\begin{align*}
E_{v}\left[\Phi\right] =  \braket{\Phi|\hat{T}|\Phi}+J\left[\rho\right]+E_{xc}[\Phi] + \int d\mathbf{r}v(\mathbf{r})\rho(\mathbf{r})  
\end{align*}
where we use $\rho(\mathbf{r})=\rho_{n,w}(\mathbf{r})=\gamma_{s}(\mathbf{r},\mathbf{r})=\braket{\Phi|\hat{\rho}(\mathbf{r})|\Phi}$,
denote the classical Coulomb energy as $J\left[\rho\right]$, and define the exchange-correlation energy functional as
\begin{align}\label{eq:Exc_Phigamma}
E_{xc}[\Phi]=\braket{\Psi_{n,w}^1|\hat{T} + V_{ee}|\Psi_{n,w}^1} -\braket{\Phi|\hat{T}|\Phi}-J\left[\rho\right] 
\end{align}
Subject to the ground-state KS assumption, there exists a ground-state wavefunction for a system with full electron interaction, $\Psi_{0,w}^1$, and these expressions for the exchange-correlation energy are equivalent to the traditional KS definition. Moreover, as the ground-state density $\rho(\mathbf{r})$ uniquely determines
$\Phi$ and $\gamma_{s}$, we can use the density as the basic variable and define the exchange-correlation energy (and, if we wish, even the noninteracting kinetic energy) as implicit functionals of $\rho(\mathbf{r})$. 

Similarly, subject to the excited-state KS assumption, there exists an excited-state wavefunction for a system with full electron interaction, $\Psi_{n,w}^1$, and $E_{xc}$ can be expressed as a functional of the noninteracting wavefunction, $\Phi$, or the noninteracting density matrix, $\gamma_{s}$. This leads to Eqs. (\ref{eq:Exc_Phigamma}), which are applicable to
both ground and excited states. The basic variable is ($n, w_s$),  $\gamma_{s}$ 
or $\Phi$; the electron density $\rho(\mathbf{r})$
no longer suffices. Thus, while ground-state KS theory is a DFT, excited-state KS theory is a $n$PFT, $\Phi$FT, or $\rho_{s}$FT. Thus, as shown in Table I, excited-state KS theory subsumes ground-state KS-DFT.

We now derive the working equations for the stationary points. For convenience, we will work directly with the set of $N$ orbitals, $\left\{ \phi_{i}(\mathbf{r})\right\} $, that are occupied in the $n_s^{\text{th}}$ excited state of the noninteracting Hamiltonian (cf. Eq. (\ref{eq:H_s})), each of which is an eigenfunction of the one-electron Hamiltonian with the trial
potential $w_{s}(\mathbf{r})$:
\begin{equation}
\left(-\tfrac{1}{2}\nabla^{2}+w_{s}(\mathbf{r})\right)\phi_{i}(\mathbf{r})=\varepsilon_{i}\phi_{i}(\mathbf{r}).\label{eq:OEP_orb}
\end{equation}
We can then express Theorem 2's stationarity condition in terms of the occupied orbitals or, equivalently, $\Phi$  or  $\gamma_{s}$. E.g., 
\begin{align}\label{eq:OEPCondition}
0=&\frac{\delta E_{v}[\gamma_{s}]}{\delta w_{s}(\mathbf{r}')}  =\sum_{i}\int d\mathbf{r}\frac{\delta E_{v}[\gamma_{s}]}{\delta\phi_i^{*}(\mathbf{r})}\frac{\delta\phi_{i}^{*}(\mathbf{r})}{\delta w_{s}(\mathbf{r}')}+c.c.\nonumber \\
  =&\sum_{i}\int d\mathbf{r}\frac{\delta\phi_{i}^{*}(\mathbf{r})}{\delta w_{s}(\mathbf{r}')}\biggl[\left(-\frac{1}{2}\nabla^{2}+v(\mathbf{r})+v_{J}\right)\phi_{i}(\mathbf{r}) \nonumber \\ &\qquad \qquad \qquad \qquad +\frac{\delta E_{xc}[\gamma_{s}(\mathbf{r},\mathbf{r}')]}{\delta\phi_{i}^{*}(\mathbf{r}')}\biggr]+c.c.
\end{align}
Here $v_J(\mathbf{r})$ denotes the classical Coulomb potential. This stationary condition is the excited-state generalization of the optimized effective potential
(OEP) minimum energy condition for the ground states.\cite{sharpVariationalApproachUnipotential1953,talmanOptimizedEffectiveAtomic1976}

When $E_{xc}[\gamma_{s}]$ is \textit{approximated} by an explicit functional of the electron density, $E_{xc}[\rho(\mathbf{r})]$, as in the local density functional approximation or a generalized gradient approximation, we have
\begin{equation*}
    \frac{\delta E_{xc}[\gamma_{s}]}{\delta\phi_{i}^{*}(\mathbf{r})}=\frac{\delta E_{xc}[\rho]}{\delta\phi_{i}^{*}(\mathbf{r})}=\frac{\delta E_{xc}[\rho]}{\delta\rho(\mathbf{r})}\phi_{i}(\mathbf{r})=v_{xc}(\mathbf{r})\phi_{i}(\mathbf{r}),
\end{equation*}
where $v_{xc}(\mathbf{r})=\frac{\delta E_{xc}[\rho]}{\delta\rho(\mathbf{r})}$ is the (local) exchange-correlation potential. The OEP stationary condition (\ref{eq:OEPCondition}) is then satisfied when the trial potential is equal to the Kohn-Sham potential constructed from the density composed by the orbitals occupied in the $n_s^{\text{th}}$ excited state,  
\begin{equation}
w_s(\mathbf{r}) = v_{\text{eff}}(\mathbf{r})=v(\mathbf{r})+v_{J}(\mathbf{r})+v_{xc}(\mathbf{r}).\label{eq:v_eff}
\end{equation}
The $\Delta$SCF approach is to reach the stationary condition by solving, self-consistently, the excited-state KS equations
obtained by inserting Eq. (\ref{eq:v_eff}) into Eq. (\ref{eq:OEP_orb}).

When $E_{xc}[\gamma_{s}]$ is given
as an explicit functional of $\gamma_{s}(\mathbf{r},\mathbf{r}')$,
as in Hartree-Fock and hybrid DFT, 
\begin{equation} \label{eq:genKSderivative}  
\frac{\delta E_{xc}[\Phi]}{\delta\phi_{i}^{*}(\mathbf{r})}=\frac{\delta E_{xc}[\gamma_{s}]}{\delta\phi_{i}^{*}(\mathbf{r})}=\int d\mathbf{r}'v_{xc}(\mathbf{r},\mathbf{r}')\phi_{i}(\mathbf{r}')  
\end{equation}
 where the nonlocal exchange-correlation operator is defined as $v_{xc}(\mathbf{r},\mathbf{r}')=\frac{\delta E_{xc}[\gamma_{s}]}{\delta\gamma_{s}(\mathbf{r}',\mathbf{r})}$. The noninteracting potential can be determined by substituting (\ref{eq:genKSderivative}) into the OEP equation (\ref{eq:OEPCondition}). Alternatively, a different stationary point can be obtained using a generalized Kohn-Sham (GKS) approach where one allows the noninteracting reference system described by the one-electron equations (\ref{eq:OEP_orb}) to have a nonlocal effective potential,\cite{jinGeneralizedOptimizedEffective2017} 
\begin{equation}
w_s(\mathbf{r},\mathbf{r}') = v_{\text{eff}}(\mathbf{r},\mathbf{r}')=\left(v(\mathbf{r})+v_{J}(\mathbf{r})\right)\delta(\mathbf{r}-\mathbf{r}')+v_{xc}(\mathbf{r},\mathbf{r}').\label{eq:v_eff_nonlocal}
\end{equation}
This GKS approach was already developed in the original work of Kohn and Sham and was called Hartree-Fock-Kohn-Sham
method\cite{kohnSelfconsistentEquationsIncluding1965,parrDensityFunctionalTheoryAtoms1989}. GKS calculations are much easier computationally than OEP calculations, but the difference in total energy and electron density, for ground states, is usually small. However, there is a major difference in the values of the (virtual) orbital energies and their interpretation.\cite{cohenFractionalChargePerspective2008b,garrickExactGeneralizedKohnSham2020}  

Theorem 2 establishes the foundation for $\Delta$SCF calculations
in KS, OEP and GKS calculations. As expressed in the Figure, while $\gamma_s$ or $\Phi$ are used as explicit variables in ground-state KS theory, they are inferred from the noninteracting Hamiltonian and its potential $w_s(\mathbf{r})$, which are determined uniquely (up to a trivial constant) by the ground-state density. Thus ground-state KS theory is still a \textit{bona fide} DFT. Because excited-state densities do not uniquely determine their associated potentials, excited-state KS theory cannot be based only on the density, but it can be based on the noninteracting potential. The excited-state electron density, as the functional derivative of the eigenenergy with respect to the potential, still plays a key role.\cite{jinIntroductoryLectureWhen2020} Specifically, at stationary points of our new functionals the density of the noninteracting system is equal to the density of corresponding physical system, for both ground and excited states.

\begin{table}
\label{tab:functionals}

\caption{Ground- and excited-state energy functionals and their variables. The new functionals are in the last row.}

\begin{tabular}{|c|c|c|c|c|c|}
\hline
 &  & \multicolumn{4}{c|}{Theoretical Formulations}\tabularnewline
\hline
Theory & States & DFT & $n$PFT & $\Phi$FT & $\gamma_{s}$FT\tabularnewline
\hline
\hline
HK & ground & $E_{v}[\rho(\mathbf{r})]$ & $E_{v}^{1}[0,w]$ &  & \tabularnewline
\hline
KS & ground & $\text{\ensuremath{E_{v}}[\ensuremath{\rho}(\ensuremath{\mathbf{r}})] }$ & $E_{v}[0,w_{s}]$ & $E_{v}[\Phi]$ & $E_{v}[\gamma_{s}(\mathbf{r},\mathbf{r}')]$\tabularnewline
\hline
 & all & N/A & $E_{v}[n_s,w_{s}]$ & $E_{v}[\Phi]$ & $E_{v}[\gamma_{s}(\mathbf{r},\mathbf{r}')]$\tabularnewline
\hline
\end{tabular}

\end{table}

\begin{figure}
\begin{centering}
\includegraphics[width=3.4in]{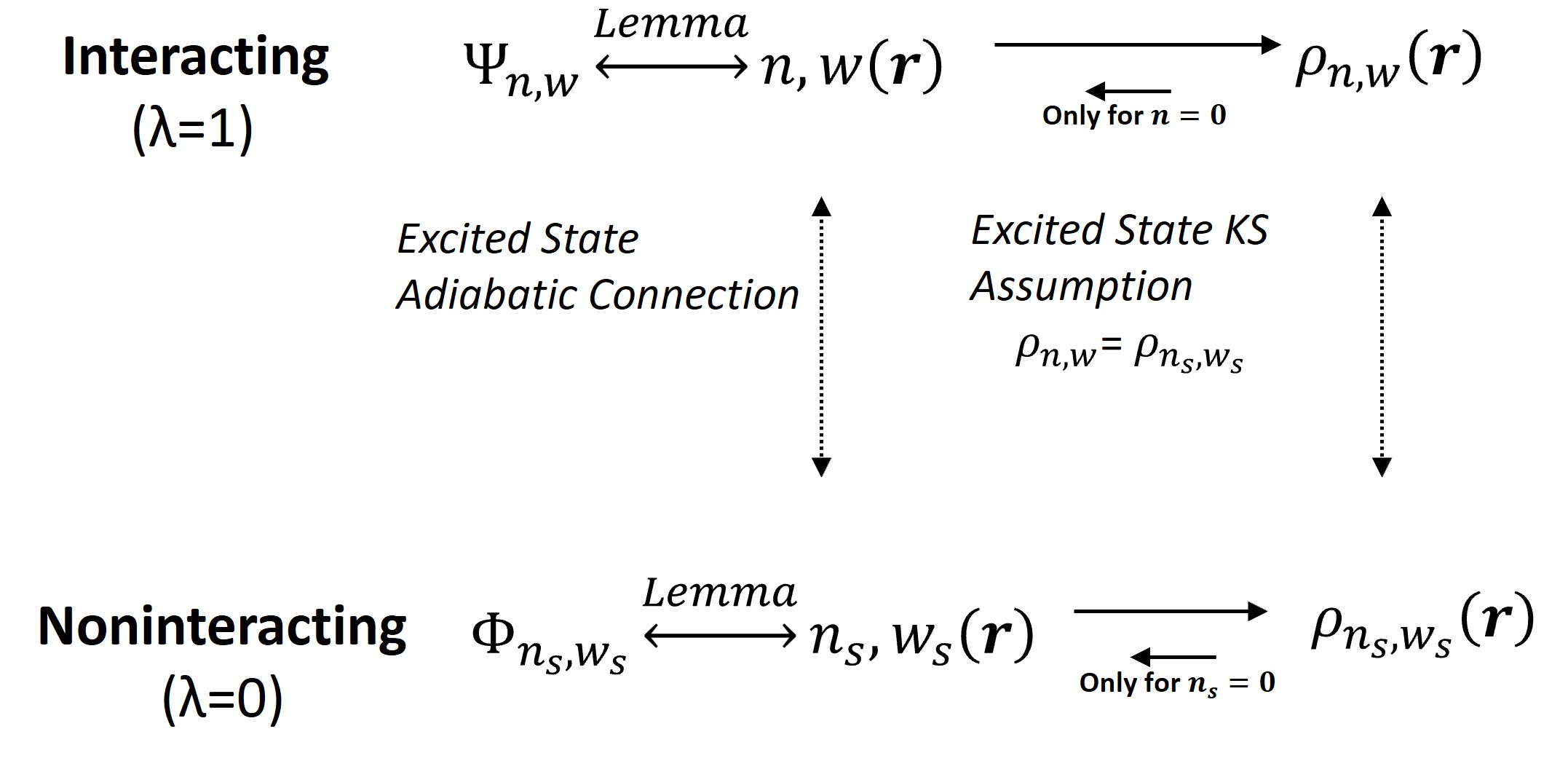}\caption{Variables and Their Relationship for Describing Excited States. }
\par\end{centering}
\end{figure}

To construct the exchange-correlation energy functional, one uses the adiabatic connection. It is simplest and most conventional
to choose a constant-density adiabatic connection, where $n^\lambda$ and $w^\lambda(\mathbf{r})$ are chosen so that the specified excited state density, $\braket{\Psi_{n^\lambda,w^\lambda}^\lambda|\hat{\rho}(\mathbf{r})|\Psi_{n^\lambda,w^\lambda}^\lambda}$, 
of the adiabatic connection Hamiltonian (cf. Eq. (\ref{eq:H_w})) remains constant.\cite{Gorling1999esDFTgenAC} By tracing this pathway, $\left(n_{s},w_{s}\right)=\left(n^0,w^0\right)$ (equivalently $\Phi$ and $\gamma_s$) determine the properties of the relevant excited state of the interacting system, $\left(n^1,w^1\right)$, and the exchange-correlation functional can be explicitly written,
\begin{align}
E_{xc}[\gamma_{s}] & =\int_{0}^{1}d\lambda \braket{\Psi_{n^\lambda,w^\lambda}^\lambda|\hat{V}_{ee}|\Psi_{n^\lambda,w^\lambda}^\lambda} - J[\rho]
\label{eq:E_xc_AC}
\end{align}
\textit{As a functional of $\gamma_s$ or $\Phi$, Eq. (\ref{eq:E_xc_AC}) is the same for ground and
excited states.} This justifies the using approximations to $E_{xc}$ that were originally designed for ground states for excited states in the $\Delta$SCF method. We are certainly not the first to note the utility of $\Phi$- or $\gamma_s$- functionals for treating excited states, as  $\Phi$ or $\gamma_s$ are always the operational variables in $\Delta$SCF calculations.\cite{Slater1972AdvQuantumChem,SlaterWood1970SlaterTM,
ZieglerRaukBaerends1977deltaSCF,JonesGunnarsson1989dftreview,Gorling1996esDFT,Gorling1999esDFTgenAC, Gorling2000esDFTPRL,KowalczykanVoorhis2011deltaSCF,GouldKronik2021ensembleDFTgenKS}  Our work defines the corresponding energy functionals and formulates their variational principles, based on the potential-functional formulation.

For ground states, the original Hohenberg-Kohn-Sham theorems were restricted to $v$-representable densities,\cite{HohenbergKohn1964,kohnSelfconsistentEquationsIncluding1965} motivating the development of functionals defined on the broader set of $N$-representable densities.\cite{Levy1979constrainedsearch,valoneOnetoOneMappingOneParticle1980,Lieb1983dftIJQC,Lieb1985,ayersAxiomaticFormulationsHohenbergKohn2006} However, because there is no excited-state Hohenberg-Kohn theorem, $N$-representable DFT frameworks cannot be directly extended to excited states.\cite{Lieb1985,Gorling1996esDFT, Gorling1999esDFTgenAC, Gorling2000esDFTPRL,LevyNagy1999esDFTKoopmans, LevyNagy1999esDFTPRL, NagyLevy2001esDFTconstrainedsearch, RNAyersLevy2009excited, NagyLevyAyers2009esDFTchapter} Our ground-state potential functional theory preserves the $v$-representable framework of the original Hohenberg-Kohn-Sham DFT\cite{YangAyersWu2004potentialFT} and has now been extended to excited states, establishing the rigor of the $\Delta$SCF method.


Summarizing, we present three universal ground-and-excited-state functionals
$E_{v}[n,w_{s}]=E_{v}[\Phi]=E_{v}[\gamma_s]$. The first functional generalizes the ground-state PFT formulation.\cite{YangAyersWu2004potentialFT} Similarly, $E_{v}[\Phi]$ and $E_{v}[\gamma_s]$ extend the ground-state (generalized) Kohn-Sham approach. The minimum of these functionals yields the ground state energy and its density. Their other stationary points yield excited-state energies and electron densities, as shown in Theorem 2. The realization that functionals can give exact results for ground states \textit{and} excited states should be used when designing and testing new exchange-correlation functionals.

\begin{acknowledgments}
WY acknowledges support from the National Science Foundation (CHE-2154831) and the National Institute of Health (R01-GM061870). PWA acknowledges support from NSERC, the Canada Research Chairs, and the Digital Research Alliance of Canada. 
\end{acknowledgments}



\begin{thebibliography}{89}%
\makeatletter
\providecommand \@ifxundefined [1]{%
 \@ifx{#1\undefined}
}%
\providecommand \@ifnum [1]{%
 \ifnum #1\expandafter \@firstoftwo
 \else \expandafter \@secondoftwo
 \fi
}%
\providecommand \@ifx [1]{%
 \ifx #1\expandafter \@firstoftwo
 \else \expandafter \@secondoftwo
 \fi
}%
\providecommand \natexlab [1]{#1}%
\providecommand \enquote  [1]{``#1''}%
\providecommand \bibnamefont  [1]{#1}%
\providecommand \bibfnamefont [1]{#1}%
\providecommand \citenamefont [1]{#1}%
\providecommand \href@noop [0]{\@secondoftwo}%
\providecommand \href [0]{\begingroup \@sanitize@url \@href}%
\providecommand \@href[1]{\@@startlink{#1}\@@href}%
\providecommand \@@href[1]{\endgroup#1\@@endlink}%
\providecommand \@sanitize@url [0]{\catcode `\\12\catcode `\$12\catcode
  `\&12\catcode `\#12\catcode `\^12\catcode `\_12\catcode `\%12\relax}%
\providecommand \@@startlink[1]{}%
\providecommand \@@endlink[0]{}%
\providecommand \url  [0]{\begingroup\@sanitize@url \@url }%
\providecommand \@url [1]{\endgroup\@href {#1}{\urlprefix }}%
\providecommand \urlprefix  [0]{URL }%
\providecommand \Eprint [0]{\href }%
\providecommand \doibase [0]{https://doi.org/}%
\providecommand \selectlanguage [0]{\@gobble}%
\providecommand \bibinfo  [0]{\@secondoftwo}%
\providecommand \bibfield  [0]{\@secondoftwo}%
\providecommand \translation [1]{[#1]}%
\providecommand \BibitemOpen [0]{}%
\providecommand \bibitemStop [0]{}%
\providecommand \bibitemNoStop [0]{.\EOS\space}%
\providecommand \EOS [0]{\spacefactor3000\relax}%
\providecommand \BibitemShut  [1]{\csname bibitem#1\endcsname}%
\let\auto@bib@innerbib\@empty
\bibitem [{\citenamefont {Hohenberg}\ and\ \citenamefont
  {Kohn}(1964)}]{HohenbergKohn1964}%
  \BibitemOpen
  \bibfield  {author} {\bibinfo {author} {\bibfnamefont {P.}~\bibnamefont
  {Hohenberg}}\ and\ \bibinfo {author} {\bibfnamefont {W.}~\bibnamefont
  {Kohn}},\ }\bibfield  {title} {\bibinfo {title} {Inhomogeneous electron
  gas},\ }\href@noop {} {\bibfield  {journal} {\bibinfo  {journal} {Phys.Rev.}\
  }\textbf {\bibinfo {volume} {136}},\ \bibinfo {pages} {B864} (\bibinfo {year}
  {1964})}\BibitemShut {NoStop}%
\bibitem [{\citenamefont {Kohn}\ and\ \citenamefont
  {Sham}(1965)}]{kohnSelfconsistentEquationsIncluding1965}%
  \BibitemOpen
  \bibfield  {author} {\bibinfo {author} {\bibfnamefont {W.}~\bibnamefont
  {Kohn}}\ and\ \bibinfo {author} {\bibfnamefont {L.~J.}\ \bibnamefont
  {Sham}},\ }\bibfield  {title} {\bibinfo {title} {Self-consistent equations
  including exchange and correlation effects},\ }\href@noop {} {\bibfield
  {journal} {\bibinfo  {journal} {Phys.Rev.}\ }\textbf {\bibinfo {volume}
  {140}},\ \bibinfo {pages} {A1133} (\bibinfo {year} {1965})}\BibitemShut
  {NoStop}%
\bibitem [{\citenamefont {Levy}(1979)}]{Levy1979constrainedsearch}%
  \BibitemOpen
  \bibfield  {author} {\bibinfo {author} {\bibfnamefont {M.}~\bibnamefont
  {Levy}},\ }\bibfield  {title} {\bibinfo {title} {Universal variational
  functionals of electron-densities, 1st- order density-matrices, and natural
  spin-orbitals and solution of the v-representability problem},\ }\href@noop
  {} {\bibfield  {journal} {\bibinfo  {journal} {Proceedings of the National
  Academy of Sciences}\ }\textbf {\bibinfo {volume} {76}},\ \bibinfo {pages}
  {6062} (\bibinfo {year} {1979})}\BibitemShut {NoStop}%
\bibitem [{\citenamefont
  {Valone}(1980)}]{valoneOnetoOneMappingOneParticle1980}%
  \BibitemOpen
  \bibfield  {author} {\bibinfo {author} {\bibfnamefont {S.~M.}\ \bibnamefont
  {Valone}},\ }\bibfield  {title} {\bibinfo {title} {A {{One-to-One Mapping
  Between One-Particle Densities}} and {{Some Normal-Particle Ensembles}}},\
  }\href {https://doi.org/10.1063/1.440656} {\bibfield  {journal} {\bibinfo
  {journal} {Journal of Chemical Physics}\ }\textbf {\bibinfo {volume} {73}},\
  \bibinfo {pages} {4653} (\bibinfo {year} {1980})}\BibitemShut {NoStop}%
\bibitem [{\citenamefont {Lieb}(1983)}]{Lieb1983dftIJQC}%
  \BibitemOpen
  \bibfield  {author} {\bibinfo {author} {\bibfnamefont {E.~H.}\ \bibnamefont
  {Lieb}},\ }\bibfield  {title} {\bibinfo {title} {Density functionals for
  coulomb systems},\ }\href@noop {} {\bibfield  {journal} {\bibinfo  {journal}
  {International Journal of Quantum Chemistry}\ }\textbf {\bibinfo {volume}
  {24}},\ \bibinfo {pages} {243} (\bibinfo {year} {1983})}\BibitemShut
  {NoStop}%
\bibitem [{\citenamefont {Lieb}(1985)}]{Lieb1985}%
  \BibitemOpen
  \bibfield  {author} {\bibinfo {author} {\bibfnamefont {E.~H.}\ \bibnamefont
  {Lieb}},\ }\bibfield  {title} {\bibinfo {title} {Density functionals for
  coulomb systems},\ }\href@noop {} {\bibfield  {journal} {\bibinfo  {journal}
  {NATO ASI Series, Series B}\ }\textbf {\bibinfo {volume} {123}},\ \bibinfo
  {pages} {31} (\bibinfo {year} {1985})}\BibitemShut {NoStop}%
\bibitem [{\citenamefont
  {Ayers}(2006)}]{ayersAxiomaticFormulationsHohenbergKohn2006}%
  \BibitemOpen
  \bibfield  {author} {\bibinfo {author} {\bibfnamefont {{\relax
  PW}.}~\bibnamefont {Ayers}},\ }\bibfield  {title} {\bibinfo {title}
  {Axiomatic formulations of the {{Hohenberg-Kohn}} functional},\ }\bibfield
  {journal} {\bibinfo  {journal} {Physical Review a}\ }\textbf {\bibinfo
  {volume} {73}},\ \href {https://doi.org/10.1103/PhysRevA.73.012513}
  {10.1103/PhysRevA.73.012513} (\bibinfo {year} {2006})\BibitemShut {NoStop}%
\bibitem [{\citenamefont {Teale}\ \emph {et~al.}(2022)\citenamefont {Teale},
  \citenamefont {Helgaker}, \citenamefont {Savin}, \citenamefont {Adano},
  \citenamefont {Aradi}, \citenamefont {Arbuznikov}, \citenamefont {Ayers},
  \citenamefont {Baerends}, \citenamefont {Barone}, \citenamefont {Calaminici}
  \emph {et~al.}}]{TealeHelgaker2022DFTexchange}%
  \BibitemOpen
  \bibfield  {author} {\bibinfo {author} {\bibfnamefont {A.~M.}\ \bibnamefont
  {Teale}}, \bibinfo {author} {\bibfnamefont {T.}~\bibnamefont {Helgaker}},
  \bibinfo {author} {\bibfnamefont {A.}~\bibnamefont {Savin}}, \bibinfo
  {author} {\bibfnamefont {C.}~\bibnamefont {Adano}}, \bibinfo {author}
  {\bibfnamefont {B.}~\bibnamefont {Aradi}}, \bibinfo {author} {\bibfnamefont
  {A.~V.}\ \bibnamefont {Arbuznikov}}, \bibinfo {author} {\bibfnamefont
  {P.}~\bibnamefont {Ayers}}, \bibinfo {author} {\bibfnamefont {E.~J.}\
  \bibnamefont {Baerends}}, \bibinfo {author} {\bibfnamefont {V.}~\bibnamefont
  {Barone}}, \bibinfo {author} {\bibfnamefont {P.}~\bibnamefont {Calaminici}},
  \emph {et~al.},\ }\bibfield  {title} {\bibinfo {title} {Dft exchange: sharing
  perspectives on the workhorse of quantum chemistry and materials science},\
  }\href@noop {} {\bibfield  {journal} {\bibinfo  {journal} {Physical Chemistry
  Chemical Physics}\ } (\bibinfo {year} {2022})},\ \bibinfo {note}
  {{\'a}}\BibitemShut {NoStop}%
\bibitem [{\citenamefont {Von~Barth}\ and\ \citenamefont
  {Hedin}(1972)}]{vonBarthHedin1972SDFT}%
  \BibitemOpen
  \bibfield  {author} {\bibinfo {author} {\bibfnamefont {U.}~\bibnamefont
  {Von~Barth}}\ and\ \bibinfo {author} {\bibfnamefont {L.}~\bibnamefont
  {Hedin}},\ }\bibfield  {title} {\bibinfo {title} {A local
  exchange-correlation potential for the spin polarized case i},\ }\href@noop
  {} {\bibfield  {journal} {\bibinfo  {journal} {J. Phys. C}\ }\textbf
  {\bibinfo {volume} {5}},\ \bibinfo {pages} {1629} (\bibinfo {year}
  {1972})}\BibitemShut {NoStop}%
\bibitem [{\citenamefont {Gunnarsson}\ and\ \citenamefont
  {Lundqvist}(1976)}]{GunnarssonLundqvist1976}%
  \BibitemOpen
  \bibfield  {author} {\bibinfo {author} {\bibfnamefont {O.}~\bibnamefont
  {Gunnarsson}}\ and\ \bibinfo {author} {\bibfnamefont {B.~I.}\ \bibnamefont
  {Lundqvist}},\ }\bibfield  {title} {\bibinfo {title} {Exchange and
  correlation in atoms, molecules and solids by the spin-density-functional
  formalism},\ }\href@noop {} {\bibfield  {journal} {\bibinfo  {journal}
  {Physical Review B}\ }\textbf {\bibinfo {volume} {13}},\ \bibinfo {pages}
  {4274} (\bibinfo {year} {1976})}\BibitemShut {NoStop}%
\bibitem [{\citenamefont {Perdew}\ and\ \citenamefont
  {Zunger}(1981)}]{PerdewZunger1981spinDFT}%
  \BibitemOpen
  \bibfield  {author} {\bibinfo {author} {\bibfnamefont {J.~P.}\ \bibnamefont
  {Perdew}}\ and\ \bibinfo {author} {\bibfnamefont {A.}~\bibnamefont
  {Zunger}},\ }\bibfield  {title} {\bibinfo {title} {Self-interaction
  correction to density-functional approximations for many-electron systems},\
  }\href@noop {} {\bibfield  {journal} {\bibinfo  {journal} {Physical Review
  B}\ }\textbf {\bibinfo {volume} {23}},\ \bibinfo {pages} {5048} (\bibinfo
  {year} {1981})}\BibitemShut {NoStop}%
\bibitem [{\citenamefont {Gorling}(1993)}]{Gorling1993SymmetryDFT}%
  \BibitemOpen
  \bibfield  {author} {\bibinfo {author} {\bibfnamefont {A.}~\bibnamefont
  {Gorling}},\ }\bibfield  {title} {\bibinfo {title} {Symmetry in
  density-functional theory},\ }\href@noop {} {\bibfield  {journal} {\bibinfo
  {journal} {Physical Review A}\ }\textbf {\bibinfo {volume} {47}},\ \bibinfo
  {pages} {2783} (\bibinfo {year} {1993})}\BibitemShut {NoStop}%
\bibitem [{\citenamefont {Theophilou}(1997)}]{Theophilou1997ensembleDFT}%
  \BibitemOpen
  \bibfield  {author} {\bibinfo {author} {\bibfnamefont {A.~K.}\ \bibnamefont
  {Theophilou}},\ }\bibfield  {title} {\bibinfo {title} {Density functional
  theory for excited states and special symmetries},\ }\href@noop {} {\bibfield
   {journal} {\bibinfo  {journal} {International Journal of Quantum Chemistry}\
  }\textbf {\bibinfo {volume} {61}},\ \bibinfo {pages} {333} (\bibinfo {year}
  {1997})}\BibitemShut {NoStop}%
\bibitem [{\citenamefont {Theophilou}(1998)}]{Theophilou1998essymmetries}%
  \BibitemOpen
  \bibfield  {author} {\bibinfo {author} {\bibfnamefont {A.~K.}\ \bibnamefont
  {Theophilou}},\ }\bibfield  {title} {\bibinfo {title} {Rigorous formulation
  of a kohn and sham theory for states with special symmetries},\ }\href@noop
  {} {\bibfield  {journal} {\bibinfo  {journal} {International Journal of
  Quantum Chemistry}\ }\textbf {\bibinfo {volume} {69}},\ \bibinfo {pages}
  {461} (\bibinfo {year} {1998})},\ \bibinfo {note} {0020-7608}\BibitemShut
  {NoStop}%
\bibitem [{\citenamefont {Gaudoin}\ and\ \citenamefont
  {Burke}(2004)}]{GaudoinBurke2004esHKtheorem}%
  \BibitemOpen
  \bibfield  {author} {\bibinfo {author} {\bibfnamefont {R.}~\bibnamefont
  {Gaudoin}}\ and\ \bibinfo {author} {\bibfnamefont {K.}~\bibnamefont
  {Burke}},\ }\bibfield  {title} {\bibinfo {title} {Lack of hohenberg-kohn
  theorem for excited states},\ }\href@noop {} {\bibfield  {journal} {\bibinfo
  {journal} {Physical Review Letters}\ }\textbf {\bibinfo {volume} {93}},\
  \bibinfo {pages} {173001} (\bibinfo {year} {2004})}\BibitemShut {NoStop}%
\bibitem [{\citenamefont {Gaudoin}\ and\ \citenamefont
  {Burke}(2005)}]{GaudoinBurke2005esHKtheoremerrata}%
  \BibitemOpen
  \bibfield  {author} {\bibinfo {author} {\bibfnamefont {R.}~\bibnamefont
  {Gaudoin}}\ and\ \bibinfo {author} {\bibfnamefont {K.}~\bibnamefont
  {Burke}},\ }\bibfield  {title} {\bibinfo {title} {Lack of hohenberg-kohn
  theorem for excited states (vol 93, art no 073001, 2004)},\ }\href@noop {}
  {\bibfield  {journal} {\bibinfo  {journal} {Physical Review Letters}\
  }\textbf {\bibinfo {volume} {94}},\ \bibinfo {pages} {029901} (\bibinfo
  {year} {2005})}\BibitemShut {NoStop}%
\bibitem [{\citenamefont {Samal}\ and\ \citenamefont
  {Harbola}(2006)}]{SamalHarbola2006esHK}%
  \BibitemOpen
  \bibfield  {author} {\bibinfo {author} {\bibfnamefont {P.}~\bibnamefont
  {Samal}}\ and\ \bibinfo {author} {\bibfnamefont {M.~K.}\ \bibnamefont
  {Harbola}},\ }\bibfield  {title} {\bibinfo {title} {Exploring foundations of
  time-independent density functional theory for excited states},\ }\href@noop
  {} {\bibfield  {journal} {\bibinfo  {journal} {Journal of Physics B}\
  }\textbf {\bibinfo {volume} {39}},\ \bibinfo {pages} {4065} (\bibinfo {year}
  {2006})}\BibitemShut {NoStop}%
\bibitem [{\citenamefont {Samal}\ \emph {et~al.}(2006)\citenamefont {Samal},
  \citenamefont {Harbola},\ and\ \citenamefont
  {Holas}}]{SamalHarbolaHolas2006esHK}%
  \BibitemOpen
  \bibfield  {author} {\bibinfo {author} {\bibfnamefont {P.}~\bibnamefont
  {Samal}}, \bibinfo {author} {\bibfnamefont {M.~K.}\ \bibnamefont {Harbola}},\
  and\ \bibinfo {author} {\bibfnamefont {A.}~\bibnamefont {Holas}},\ }\bibfield
   {title} {\bibinfo {title} {Density-to-potential map in time-independent
  excited-state density-functional theory},\ }\href@noop {} {\bibfield
  {journal} {\bibinfo  {journal} {Chemical Physics Letters}\ }\textbf {\bibinfo
  {volume} {419}},\ \bibinfo {pages} {217} (\bibinfo {year}
  {2006})}\BibitemShut {NoStop}%
\bibitem [{\citenamefont {Gorling}(1996)}]{Gorling1996esDFT}%
  \BibitemOpen
  \bibfield  {author} {\bibinfo {author} {\bibfnamefont {A.}~\bibnamefont
  {Gorling}},\ }\bibfield  {title} {\bibinfo {title} {Density-functional theory
  for excited states},\ }\href@noop {} {\bibfield  {journal} {\bibinfo
  {journal} {Physical Review A}\ }\textbf {\bibinfo {volume} {54}},\ \bibinfo
  {pages} {3912} (\bibinfo {year} {1996})},\ \bibinfo {note}
  {1050-2947}\BibitemShut {NoStop}%
\bibitem [{\citenamefont {Gorling}(1999)}]{Gorling1999esDFTgenAC}%
  \BibitemOpen
  \bibfield  {author} {\bibinfo {author} {\bibfnamefont {A.}~\bibnamefont
  {Gorling}},\ }\bibfield  {title} {\bibinfo {title} {Density-functional theory
  beyond the hohenberg-kohn theorem},\ }\href@noop {} {\bibfield  {journal}
  {\bibinfo  {journal} {Physical Review A}\ }\textbf {\bibinfo {volume} {59}},\
  \bibinfo {pages} {3359} (\bibinfo {year} {1999})}\BibitemShut {NoStop}%
\bibitem [{\citenamefont {Gorling}(2000)}]{Gorling2000esDFTPRL}%
  \BibitemOpen
  \bibfield  {author} {\bibinfo {author} {\bibfnamefont {A.}~\bibnamefont
  {Gorling}},\ }\bibfield  {title} {\bibinfo {title} {Proper treatment of
  symmetries and excited states in a computationally tractable kohn-sham
  method},\ }\href {<Go to ISI>://000165247700009} {\bibfield  {journal}
  {\bibinfo  {journal} {Physical Review Letters}\ }\textbf {\bibinfo {volume}
  {85}},\ \bibinfo {pages} {4229} (\bibinfo {year} {2000})}\BibitemShut
  {NoStop}%
\bibitem [{\citenamefont {Ayers}\ and\ \citenamefont
  {Levy}(2009)}]{RNAyersLevy2009excited}%
  \BibitemOpen
  \bibfield  {author} {\bibinfo {author} {\bibfnamefont {P.~W.}\ \bibnamefont
  {Ayers}}\ and\ \bibinfo {author} {\bibfnamefont {M.}~\bibnamefont {Levy}},\
  }\bibfield  {title} {\bibinfo {title} {Time-independent (static)
  density-functional theories for pure excited states: Extensions and
  unification},\ }\href@noop {} {\bibfield  {journal} {\bibinfo  {journal}
  {Physical Review A}\ }\textbf {\bibinfo {volume} {80}},\ \bibinfo {pages}
  {012508} (\bibinfo {year} {2009})},\ \bibinfo {note} {1050-2947}\BibitemShut
  {NoStop}%
\bibitem [{\citenamefont {Levy}\ and\ \citenamefont
  {Nagy}(1999{\natexlab{a}})}]{LevyNagy1999esDFTKoopmans}%
  \BibitemOpen
  \bibfield  {author} {\bibinfo {author} {\bibfnamefont {M.}~\bibnamefont
  {Levy}}\ and\ \bibinfo {author} {\bibfnamefont {A.}~\bibnamefont {Nagy}},\
  }\bibfield  {title} {\bibinfo {title} {Excited-state koopmans theorem for
  ensembles},\ }\href@noop {} {\bibfield  {journal} {\bibinfo  {journal}
  {Physical Review A}\ }\textbf {\bibinfo {volume} {59}},\ \bibinfo {pages}
  {1687} (\bibinfo {year} {1999}{\natexlab{a}})}\BibitemShut {NoStop}%
\bibitem [{\citenamefont {Levy}\ and\ \citenamefont
  {Nagy}(1999{\natexlab{b}})}]{LevyNagy1999esDFTPRL}%
  \BibitemOpen
  \bibfield  {author} {\bibinfo {author} {\bibfnamefont {M.}~\bibnamefont
  {Levy}}\ and\ \bibinfo {author} {\bibfnamefont {A.}~\bibnamefont {Nagy}},\
  }\bibfield  {title} {\bibinfo {title} {Variational density-functional theory
  for an individual excited state},\ }\href@noop {} {\bibfield  {journal}
  {\bibinfo  {journal} {Physical Review Letters}\ }\textbf {\bibinfo {volume}
  {83}},\ \bibinfo {pages} {4361} (\bibinfo {year}
  {1999}{\natexlab{b}})}\BibitemShut {NoStop}%
\bibitem [{\citenamefont {Nagy}\ and\ \citenamefont
  {Levy}(2001)}]{NagyLevy2001esDFTconstrainedsearch}%
  \BibitemOpen
  \bibfield  {author} {\bibinfo {author} {\bibfnamefont {A.}~\bibnamefont
  {Nagy}}\ and\ \bibinfo {author} {\bibfnamefont {M.}~\bibnamefont {Levy}},\
  }\bibfield  {title} {\bibinfo {title} {Variational density-functional theory
  for degenerate excited states},\ }\href@noop {} {\bibfield  {journal}
  {\bibinfo  {journal} {Physical Review A}\ }\textbf {\bibinfo {volume} {63}},\
  \bibinfo {pages} {052502} (\bibinfo {year} {2001})}\BibitemShut {NoStop}%
\bibitem [{\citenamefont {Nagy}\ \emph {et~al.}(2009)\citenamefont {Nagy},
  \citenamefont {Levy},\ and\ \citenamefont
  {Ayers}}]{NagyLevyAyers2009esDFTchapter}%
  \BibitemOpen
  \bibfield  {author} {\bibinfo {author} {\bibfnamefont {A.}~\bibnamefont
  {Nagy}}, \bibinfo {author} {\bibfnamefont {M.}~\bibnamefont {Levy}},\ and\
  \bibinfo {author} {\bibfnamefont {P.~W.}\ \bibnamefont {Ayers}},\ }\bibinfo
  {title} {Time-independent theory for a single excited state},\ in\ \href@noop
  {} {\emph {\bibinfo {booktitle} {Chemical reactivity theory: A density
  functional view}}},\ \bibinfo {editor} {edited by\ \bibinfo {editor}
  {\bibfnamefont {P.~K.}\ \bibnamefont {Chattaraj}}}\ (\bibinfo  {publisher}
  {Taylor and Francis},\ \bibinfo {address} {Boca Raton},\ \bibinfo {year}
  {2009})\ p.\ \bibinfo {pages} {121}\BibitemShut {NoStop}%
\bibitem [{\citenamefont {Harbola}\ and\ \citenamefont
  {Samal}(2009)}]{HarbolaSamal2009esHK}%
  \BibitemOpen
  \bibfield  {author} {\bibinfo {author} {\bibfnamefont {M.~K.}\ \bibnamefont
  {Harbola}}\ and\ \bibinfo {author} {\bibfnamefont {P.}~\bibnamefont
  {Samal}},\ }\bibfield  {title} {\bibinfo {title} {Time-independent
  excited-state density functional theory: study of 1s(2)2p(3)(s-4) and
  1s(2)2p(3)(d-2) states of the boron isoelectronic series up to ne5+},\
  }\href@noop {} {\bibfield  {journal} {\bibinfo  {journal} {Journal of Physics
  B-Atomic Molecular and Optical Physics}\ }\textbf {\bibinfo {volume} {42}},\
  \bibinfo {pages} {015003} (\bibinfo {year} {2009})},\ \bibinfo {note}
  {0953-4075}\BibitemShut {NoStop}%
\bibitem [{\citenamefont {Harbola}\ \emph {et~al.}(2009)\citenamefont
  {Harbola}, \citenamefont {Shamim}, \citenamefont {Samal}, \citenamefont
  {Rahaman}, \citenamefont {Ganguly},\ and\ \citenamefont
  {Mookerjee}}]{HarbolaSamal2009esHKchapter}%
  \BibitemOpen
  \bibfield  {author} {\bibinfo {author} {\bibfnamefont {M.~K.}\ \bibnamefont
  {Harbola}}, \bibinfo {author} {\bibfnamefont {M.}~\bibnamefont {Shamim}},
  \bibinfo {author} {\bibfnamefont {P.}~\bibnamefont {Samal}}, \bibinfo
  {author} {\bibfnamefont {M.}~\bibnamefont {Rahaman}}, \bibinfo {author}
  {\bibfnamefont {S.}~\bibnamefont {Ganguly}},\ and\ \bibinfo {author}
  {\bibfnamefont {A.}~\bibnamefont {Mookerjee}},\ }\bibinfo {title}
  {Time-independent excited-state density functional theory},\ in\ \href@noop
  {} {\emph {\bibinfo {booktitle} {Computational Methods in Science and
  Engineering, Vol 1 - Advances in Computational Science}}},\ \bibinfo {series}
  {Aip Conference Proceedings}, Vol.\ \bibinfo {volume} {1108},\ \bibinfo
  {editor} {edited by\ \bibinfo {editor} {\bibfnamefont {G.}~\bibnamefont
  {Maroulis}}\ and\ \bibinfo {editor} {\bibfnamefont {T.~E.}\ \bibnamefont
  {Simos}}}\ (\bibinfo {year} {2009})\ pp.\ \bibinfo {pages} {54--70},\
  \bibinfo {note} {0094-243X 978-0-7354-0644-5}\BibitemShut {NoStop}%
\bibitem [{\citenamefont {Theophilou}(1979)}]{Theophilou1979ensembleDFT}%
  \BibitemOpen
  \bibfield  {author} {\bibinfo {author} {\bibfnamefont {A.~K.}\ \bibnamefont
  {Theophilou}},\ }\bibfield  {title} {\bibinfo {title} {Energy density
  functional formalism for excited-states},\ }\href@noop {} {\bibfield
  {journal} {\bibinfo  {journal} {Journal of Physics C}\ }\textbf {\bibinfo
  {volume} {12}},\ \bibinfo {pages} {5419} (\bibinfo {year}
  {1979})}\BibitemShut {NoStop}%
\bibitem [{\citenamefont {Gross}\ \emph {et~al.}(1988)\citenamefont {Gross},
  \citenamefont {Oliveira},\ and\ \citenamefont
  {Kohn}}]{GrossOliveiraKohn1988ensembleDFT}%
  \BibitemOpen
  \bibfield  {author} {\bibinfo {author} {\bibfnamefont {E.~K.~U.}\
  \bibnamefont {Gross}}, \bibinfo {author} {\bibfnamefont {L.~N.}\ \bibnamefont
  {Oliveira}},\ and\ \bibinfo {author} {\bibfnamefont {W.}~\bibnamefont
  {Kohn}},\ }\bibfield  {title} {\bibinfo {title} {Rayleigh-ritz variational
  principle for ensembles of fractionally occupied states},\ }\href
  {ISI:A1988N059300006} {\bibfield  {journal} {\bibinfo  {journal} {Physical
  Review A}\ ,\ \bibinfo {pages} {2805}} (\bibinfo {year} {1988})},\ \bibinfo
  {note} {journal APR 15 N0593 PHYS REV A NOT IN FILE}\BibitemShut {NoStop}%
\bibitem [{\citenamefont {Oliveira}\ \emph {et~al.}(1988)\citenamefont
  {Oliveira}, \citenamefont {Gross},\ and\ \citenamefont
  {Kohn}}]{OliveiraGrossKohn1988ensembleDFT}%
  \BibitemOpen
  \bibfield  {author} {\bibinfo {author} {\bibfnamefont {L.~N.}\ \bibnamefont
  {Oliveira}}, \bibinfo {author} {\bibfnamefont {E.~K.~U.}\ \bibnamefont
  {Gross}},\ and\ \bibinfo {author} {\bibfnamefont {W.}~\bibnamefont {Kohn}},\
  }\bibfield  {title} {\bibinfo {title} {Density-functional theory for
  ensembles of fractionally occupied states .2. application to the he atom},\
  }\href@noop {} {\bibfield  {journal} {\bibinfo  {journal} {Physical Review
  A}\ }\textbf {\bibinfo {volume} {37}},\ \bibinfo {pages} {2821} (\bibinfo
  {year} {1988})}\BibitemShut {NoStop}%
\bibitem [{\citenamefont {Oliveira}\ \emph {et~al.}(1990)\citenamefont
  {Oliveira}, \citenamefont {Gross},\ and\ \citenamefont
  {Kohn}}]{OliveiraGrossKohn1990ensembleDFT}%
  \BibitemOpen
  \bibfield  {author} {\bibinfo {author} {\bibfnamefont {L.~N.}\ \bibnamefont
  {Oliveira}}, \bibinfo {author} {\bibfnamefont {E.~K.~U.}\ \bibnamefont
  {Gross}},\ and\ \bibinfo {author} {\bibfnamefont {W.}~\bibnamefont {Kohn}},\
  }\bibfield  {title} {\bibinfo {title} {Ensemble-density functional theory for
  excited-states},\ }\href@noop {} {\bibfield  {journal} {\bibinfo  {journal}
  {International Journal of Quantum Chemistry}\ }\textbf {\bibinfo {volume}
  {S24}},\ \bibinfo {pages} {707} (\bibinfo {year} {1990})}\BibitemShut
  {NoStop}%
\bibitem [{\citenamefont {Theophilou}\ and\ \citenamefont
  {Gidopoulos}(1995)}]{TheophilouGidopoulos1995ensembleDFT}%
  \BibitemOpen
  \bibfield  {author} {\bibinfo {author} {\bibfnamefont {A.~K.}\ \bibnamefont
  {Theophilou}}\ and\ \bibinfo {author} {\bibfnamefont {N.~I.}\ \bibnamefont
  {Gidopoulos}},\ }\bibfield  {title} {\bibinfo {title} {Density-functional
  theory for excited-states},\ }\href@noop {} {\bibfield  {journal} {\bibinfo
  {journal} {International Journal of Quantum Chemistry}\ }\textbf {\bibinfo
  {volume} {56}},\ \bibinfo {pages} {333} (\bibinfo {year} {1995})}\BibitemShut
  {NoStop}%
\bibitem [{\citenamefont {Gidopoulos}\ \emph
  {et~al.}(2002{\natexlab{a}})\citenamefont {Gidopoulos}, \citenamefont
  {Papaconstantinou},\ and\ \citenamefont {Gross}}]{Gidopoulos2002ensembleDFT}%
  \BibitemOpen
  \bibfield  {author} {\bibinfo {author} {\bibfnamefont {N.~I.}\ \bibnamefont
  {Gidopoulos}}, \bibinfo {author} {\bibfnamefont {P.~G.}\ \bibnamefont
  {Papaconstantinou}},\ and\ \bibinfo {author} {\bibfnamefont {E.~K.~U.}\
  \bibnamefont {Gross}},\ }\bibfield  {title} {\bibinfo {title} {Spurious
  interactions, and their correction, in the ensemble-kohn-sham scheme for
  excited states},\ }\href@noop {} {\bibfield  {journal} {\bibinfo  {journal}
  {Physical Review Letters}\ }\textbf {\bibinfo {volume} {88}} (\bibinfo {year}
  {2002}{\natexlab{a}})}\BibitemShut {NoStop}%
\bibitem [{\citenamefont {Gidopoulos}\ \emph
  {et~al.}(2002{\natexlab{b}})\citenamefont {Gidopoulos}, \citenamefont
  {Papaconstantinou},\ and\ \citenamefont
  {Gross}}]{GidopoulosGross2002ensembleDFT}%
  \BibitemOpen
  \bibfield  {author} {\bibinfo {author} {\bibfnamefont {N.~I.}\ \bibnamefont
  {Gidopoulos}}, \bibinfo {author} {\bibfnamefont {P.~G.}\ \bibnamefont
  {Papaconstantinou}},\ and\ \bibinfo {author} {\bibfnamefont {E.~K.~U.}\
  \bibnamefont {Gross}},\ }\bibfield  {title} {\bibinfo {title} {Spurious
  interactions, and their correction, in the ensemble-kohn-sham scheme for
  excited states},\ }\href {https://doi.org/10.1103/PhysRevLett.88.033003}
  {\bibfield  {journal} {\bibinfo  {journal} {Physical Review Letters}\
  }\textbf {\bibinfo {volume} {88}},\ \bibinfo {pages} {033003} (\bibinfo
  {year} {2002}{\natexlab{b}})},\ \bibinfo {note} {pRL}\BibitemShut {NoStop}%
\bibitem [{\citenamefont {Franck}\ and\ \citenamefont
  {Fromager}(2014)}]{FranckFromager2014ensembleDFT}%
  \BibitemOpen
  \bibfield  {author} {\bibinfo {author} {\bibfnamefont {O.}~\bibnamefont
  {Franck}}\ and\ \bibinfo {author} {\bibfnamefont {E.}~\bibnamefont
  {Fromager}},\ }\bibfield  {title} {\bibinfo {title} {Generalised adiabatic
  connection in ensemble density-functional theory for excited states: example
  of the h-2 molecule},\ }\href@noop {} {\bibfield  {journal} {\bibinfo
  {journal} {Molecular Physics}\ }\textbf {\bibinfo {volume} {112}},\ \bibinfo
  {pages} {1684} (\bibinfo {year} {2014})}\BibitemShut {NoStop}%
\bibitem [{\citenamefont {Filatov}(2015)}]{Filatov2015ensembleDFT}%
  \BibitemOpen
  \bibfield  {author} {\bibinfo {author} {\bibfnamefont {M.}~\bibnamefont
  {Filatov}},\ }\bibfield  {title} {\bibinfo {title} {Spin-restricted
  ensemble-referenced kohn–sham method: basic principles and application to
  strongly correlated ground and excited states of molecules},\ }\href
  {https://doi.org/https://doi.org/10.1002/wcms.1209} {\bibfield  {journal}
  {\bibinfo  {journal} {WIREs Computational Molecular Science}\ }\textbf
  {\bibinfo {volume} {5}},\ \bibinfo {pages} {146} (\bibinfo {year}
  {2015})}\BibitemShut {NoStop}%
\bibitem [{\citenamefont {Cernatic}\ \emph {et~al.}(2021)\citenamefont
  {Cernatic}, \citenamefont {Senjean}, \citenamefont {Robert},\ and\
  \citenamefont {Fromager}}]{CernaticSenjeanFromager2021ensembleDFT}%
  \BibitemOpen
  \bibfield  {author} {\bibinfo {author} {\bibfnamefont {F.}~\bibnamefont
  {Cernatic}}, \bibinfo {author} {\bibfnamefont {B.}~\bibnamefont {Senjean}},
  \bibinfo {author} {\bibfnamefont {V.}~\bibnamefont {Robert}},\ and\ \bibinfo
  {author} {\bibfnamefont {E.}~\bibnamefont {Fromager}},\ }\bibfield  {title}
  {\bibinfo {title} {Ensemble density functional theory of neutral and charged
  excitations},\ }\href {https://doi.org/10.1007/s41061-021-00359-1} {\bibfield
   {journal} {\bibinfo  {journal} {Topics in Current Chemistry}\ }\textbf
  {\bibinfo {volume} {380}},\ \bibinfo {pages} {4} (\bibinfo {year}
  {2021})}\BibitemShut {NoStop}%
\bibitem [{\citenamefont {Deur}\ and\ \citenamefont
  {Fromager}(2019)}]{DeurFromager2019ensembleDFT}%
  \BibitemOpen
  \bibfield  {author} {\bibinfo {author} {\bibfnamefont {K.}~\bibnamefont
  {Deur}}\ and\ \bibinfo {author} {\bibfnamefont {E.}~\bibnamefont
  {Fromager}},\ }\bibfield  {title} {\bibinfo {title} {Ground and excited
  energy levels can be extracted exactly from a single ensemble
  density-functional theory calculation},\ }\href
  {https://doi.org/10.1063/1.5084312} {\bibfield  {journal} {\bibinfo
  {journal} {The Journal of Chemical Physics}\ }\textbf {\bibinfo {volume}
  {150}},\ \bibinfo {pages} {094106} (\bibinfo {year} {2019})}\BibitemShut
  {NoStop}%
\bibitem [{\citenamefont
  {Fromager}(2020)}]{fromagerIndividualCorrelationsEnsemble2020}%
  \BibitemOpen
  \bibfield  {author} {\bibinfo {author} {\bibfnamefont {E.}~\bibnamefont
  {Fromager}},\ }\bibfield  {title} {\bibinfo {title} {Individual
  {{Correlations}} in {{Ensemble Density Functional Theory}}: {{State-}} and
  {{Density-Driven Decompositions}} without {{Additional Kohn-Sham Systems}}},\
  }\href {https://doi.org/10.1103/PhysRevLett.124.243001} {\bibfield  {journal}
  {\bibinfo  {journal} {Physical Review Letters}\ }\textbf {\bibinfo {volume}
  {124}},\ \bibinfo {pages} {243001} (\bibinfo {year} {2020})}\BibitemShut
  {NoStop}%
\bibitem [{\citenamefont {Gould}\ and\ \citenamefont
  {Kronik}(2021)}]{GouldKronik2021ensembleDFTgenKS}%
  \BibitemOpen
  \bibfield  {author} {\bibinfo {author} {\bibfnamefont {T.}~\bibnamefont
  {Gould}}\ and\ \bibinfo {author} {\bibfnamefont {L.}~\bibnamefont {Kronik}},\
  }\bibfield  {title} {\bibinfo {title} {Ensemble generalized kohn–sham
  theory: The good, the bad, and the ugly},\ }\href
  {https://doi.org/10.1063/5.0040447} {\bibfield  {journal} {\bibinfo
  {journal} {The Journal of Chemical Physics}\ }\textbf {\bibinfo {volume}
  {154}},\ \bibinfo {pages} {094125} (\bibinfo {year} {2021})}\BibitemShut
  {NoStop}%
\bibitem [{\citenamefont {Gould}\ \emph {et~al.}(2023)\citenamefont {Gould},
  \citenamefont {Kooi}, \citenamefont {Gori-Giorgi},\ and\ \citenamefont
  {Pittalis}}]{GouldKooiGoriGiorgiPittalis2023ensembleDFT}%
  \BibitemOpen
  \bibfield  {author} {\bibinfo {author} {\bibfnamefont {T.}~\bibnamefont
  {Gould}}, \bibinfo {author} {\bibfnamefont {D.~P.}\ \bibnamefont {Kooi}},
  \bibinfo {author} {\bibfnamefont {P.}~\bibnamefont {Gori-Giorgi}},\ and\
  \bibinfo {author} {\bibfnamefont {S.}~\bibnamefont {Pittalis}},\ }\bibfield
  {title} {\bibinfo {title} {Electronic excited states in extreme limits via
  ensemble density functionals},\ }\href
  {https://doi.org/10.1103/PhysRevLett.130.106401} {\bibfield  {journal}
  {\bibinfo  {journal} {Physical Review Letters}\ }\textbf {\bibinfo {volume}
  {130}},\ \bibinfo {pages} {106401} (\bibinfo {year} {2023})},\ \bibinfo
  {note} {pRL}\BibitemShut {NoStop}%
\bibitem [{\citenamefont {Klein}\ and\ \citenamefont
  {Dreizler}(2002)}]{KleinDreizler2002esDFT}%
  \BibitemOpen
  \bibfield  {author} {\bibinfo {author} {\bibfnamefont {A.}~\bibnamefont
  {Klein}}\ and\ \bibinfo {author} {\bibfnamefont {R.~M.}\ \bibnamefont
  {Dreizler}},\ }\bibfield  {title} {\bibinfo {title} {Extension of kohn-sham
  theory to excited states by means of an off-diagonal density array},\
  }\href@noop {} {\bibfield  {journal} {\bibinfo  {journal} {Journal of Physics
  A}\ }\textbf {\bibinfo {volume} {35}},\ \bibinfo {pages} {2685} (\bibinfo
  {year} {2002})}\BibitemShut {NoStop}%
\bibitem [{\citenamefont {Gao}\ \emph {et~al.}(2016)\citenamefont {Gao},
  \citenamefont {Grofe}, \citenamefont {Ren},\ and\ \citenamefont
  {Bao}}]{GaoGrofe2016esMatrixDFT}%
  \BibitemOpen
  \bibfield  {author} {\bibinfo {author} {\bibfnamefont {J.~L.}\ \bibnamefont
  {Gao}}, \bibinfo {author} {\bibfnamefont {A.}~\bibnamefont {Grofe}}, \bibinfo
  {author} {\bibfnamefont {H.~S.}\ \bibnamefont {Ren}},\ and\ \bibinfo {author}
  {\bibfnamefont {P.}~\bibnamefont {Bao}},\ }\bibfield  {title} {\bibinfo
  {title} {Beyond kohn sham approximation: Hybrid multistate wave function and
  density functional theory},\ }\href@noop {} {\bibfield  {journal} {\bibinfo
  {journal} {Journal of Physical Chemistry Letters}\ }\textbf {\bibinfo
  {volume} {7}},\ \bibinfo {pages} {5143} (\bibinfo {year} {2016})},\ \bibinfo
  {note} {1948-7185}\BibitemShut {NoStop}%
\bibitem [{\citenamefont {Lu}\ and\ \citenamefont
  {Gao}(2022{\natexlab{a}})}]{LuGao2022matrixDFTesJCTC}%
  \BibitemOpen
  \bibfield  {author} {\bibinfo {author} {\bibfnamefont {Y.}~\bibnamefont
  {Lu}}\ and\ \bibinfo {author} {\bibfnamefont {J.}~\bibnamefont {Gao}},\
  }\bibfield  {title} {\bibinfo {title} {Fundamental variable and density
  representation in multistate dft for excited states},\ }\href
  {https://doi.org/10.1021/acs.jctc.2c00859} {\bibfield  {journal} {\bibinfo
  {journal} {Journal of Chemical Theory and Computation}\ }\textbf {\bibinfo
  {volume} {18}},\ \bibinfo {pages} {7403} (\bibinfo {year}
  {2022}{\natexlab{a}})},\ \bibinfo {note} {doi:
  10.1021/acs.jctc.2c00859}\BibitemShut {NoStop}%
\bibitem [{\citenamefont {Lu}\ and\ \citenamefont
  {Gao}(2022{\natexlab{b}})}]{LuGao2022matrixDFTes}%
  \BibitemOpen
  \bibfield  {author} {\bibinfo {author} {\bibfnamefont {Y.}~\bibnamefont
  {Lu}}\ and\ \bibinfo {author} {\bibfnamefont {J.}~\bibnamefont {Gao}},\
  }\bibfield  {title} {\bibinfo {title} {Multistate density functional theory
  of excited states},\ }\href {https://doi.org/10.1021/acs.jpclett.2c02088}
  {\bibfield  {journal} {\bibinfo  {journal} {The Journal of Physical Chemistry
  Letters}\ }\textbf {\bibinfo {volume} {13}},\ \bibinfo {pages} {7762}
  (\bibinfo {year} {2022}{\natexlab{b}})},\ \bibinfo {note} {doi:
  10.1021/acs.jpclett.2c02088}\BibitemShut {NoStop}%
\bibitem [{\citenamefont {Runge}\ and\ \citenamefont
  {Gross}(1984)}]{RungeGross1984tddft}%
  \BibitemOpen
  \bibfield  {author} {\bibinfo {author} {\bibfnamefont {E.}~\bibnamefont
  {Runge}}\ and\ \bibinfo {author} {\bibfnamefont {E.~K.~U.}\ \bibnamefont
  {Gross}},\ }\bibfield  {title} {\bibinfo {title} {Density-functional theory
  for time-dependent systems},\ }\href@noop {} {\bibfield  {journal} {\bibinfo
  {journal} {Physical Review Letters}\ }\textbf {\bibinfo {volume} {52}},\
  \bibinfo {pages} {997} (\bibinfo {year} {1984})}\BibitemShut {NoStop}%
\bibitem [{\citenamefont {Gross}\ and\ \citenamefont
  {Kohn}(1985)}]{GrossKohn1985tddft}%
  \BibitemOpen
  \bibfield  {author} {\bibinfo {author} {\bibfnamefont {E.~K.~U.}\
  \bibnamefont {Gross}}\ and\ \bibinfo {author} {\bibfnamefont
  {W.}~\bibnamefont {Kohn}},\ }\bibfield  {title} {\bibinfo {title} {Local
  density-functional theory of frequency-dependent linear response},\
  }\href@noop {} {\bibfield  {journal} {\bibinfo  {journal} {Physical Review
  Letters}\ }\textbf {\bibinfo {volume} {55}},\ \bibinfo {pages} {2850}
  (\bibinfo {year} {1985})}\BibitemShut {NoStop}%
\bibitem [{\citenamefont {Casida}\ and\ \citenamefont
  {Chong}(1995)}]{CasidaChong1995TDDFTreview}%
  \BibitemOpen
  \bibfield  {author} {\bibinfo {author} {\bibfnamefont {M.~E.}\ \bibnamefont
  {Casida}}\ and\ \bibinfo {author} {\bibfnamefont {D.~P.}\ \bibnamefont
  {Chong}},\ }\bibinfo {title} {Time-dependent density functional response
  theory for molecules},\ in\ \href@noop {} {\emph {\bibinfo {booktitle}
  {Recent Advances in Density Functional Methods. Part1.}}}\ (\bibinfo
  {publisher} {World Scientific},\ \bibinfo {address} {Singapore},\ \bibinfo
  {year} {1995})\ pp.\ \bibinfo {pages} {155--192}\BibitemShut {NoStop}%
\bibitem [{\citenamefont {Casida}(2009)}]{Casida2009reviewTDDFT}%
  \BibitemOpen
  \bibfield  {author} {\bibinfo {author} {\bibfnamefont {M.~E.}\ \bibnamefont
  {Casida}},\ }\bibfield  {title} {\bibinfo {title} {Time-dependent
  density-functional theory for molecules and molecular solids},\ }\href@noop
  {} {\bibfield  {journal} {\bibinfo  {journal} {Journal of Molecular
  Structure-Theochem}\ }\textbf {\bibinfo {volume} {914}},\ \bibinfo {pages}
  {3} (\bibinfo {year} {2009})},\ \bibinfo {note} {0166-1280}\BibitemShut
  {NoStop}%
\bibitem [{\citenamefont {Yang}\ \emph {et~al.}(2013)\citenamefont {Yang},
  \citenamefont {van Aggelen},\ and\ \citenamefont
  {Yang}}]{Yangvanaggelen2013ppRPAes}%
  \BibitemOpen
  \bibfield  {author} {\bibinfo {author} {\bibfnamefont {Y.}~\bibnamefont
  {Yang}}, \bibinfo {author} {\bibfnamefont {H.}~\bibnamefont {van Aggelen}},\
  and\ \bibinfo {author} {\bibfnamefont {W.}~\bibnamefont {Yang}},\ }\bibfield
  {title} {\bibinfo {title} {Double, rydberg and charge transfer excitations
  from pairing matrix fluctuation and particle-particle random phase
  approximation},\ }\href {https://doi.org/10.1063/1.4834875} {\bibfield
  {journal} {\bibinfo  {journal} {The Journal of Chemical Physics}\ }\textbf
  {\bibinfo {volume} {139}},\ \bibinfo {pages} {224105} (\bibinfo {year}
  {2013})}\BibitemShut {NoStop}%
\bibitem [{\citenamefont {Yang}\ \emph {et~al.}(2014)\citenamefont {Yang},
  \citenamefont {Peng}, \citenamefont {Lu},\ and\ \citenamefont
  {Yang}}]{YangPenLuYang2014ppRPAes}%
  \BibitemOpen
  \bibfield  {author} {\bibinfo {author} {\bibfnamefont {Y.}~\bibnamefont
  {Yang}}, \bibinfo {author} {\bibfnamefont {D.}~\bibnamefont {Peng}}, \bibinfo
  {author} {\bibfnamefont {J.}~\bibnamefont {Lu}},\ and\ \bibinfo {author}
  {\bibfnamefont {W.}~\bibnamefont {Yang}},\ }\bibfield  {title} {\bibinfo
  {title} {Excitation energies from particle-particle random phase
  approximation: Davidson algorithm and benchmark studies},\ }\href
  {https://doi.org/10.1063/1.4895792} {\bibfield  {journal} {\bibinfo
  {journal} {The Journal of Chemical Physics}\ }\textbf {\bibinfo {volume}
  {141}},\ \bibinfo {pages} {124104} (\bibinfo {year} {2014})}\BibitemShut
  {NoStop}%
\bibitem [{\citenamefont {Zhang}\ \emph {et~al.}(2015)\citenamefont {Zhang},
  \citenamefont {Peng}, \citenamefont {Zhang},\ and\ \citenamefont
  {Yang}}]{ZhangPengYang2015ppRPAes}%
  \BibitemOpen
  \bibfield  {author} {\bibinfo {author} {\bibfnamefont {D.}~\bibnamefont
  {Zhang}}, \bibinfo {author} {\bibfnamefont {D.}~\bibnamefont {Peng}},
  \bibinfo {author} {\bibfnamefont {P.}~\bibnamefont {Zhang}},\ and\ \bibinfo
  {author} {\bibfnamefont {W.}~\bibnamefont {Yang}},\ }\bibfield  {title}
  {\bibinfo {title} {Analytic gradients, geometry optimization and excited
  state potential energy surfaces from the particle-particle random phase
  approximation},\ }\href {https://doi.org/10.1039/C4CP04109G} {\bibfield
  {journal} {\bibinfo  {journal} {Physical Chemistry Chemical Physics}\
  }\textbf {\bibinfo {volume} {17}},\ \bibinfo {pages} {1025} (\bibinfo {year}
  {2015})}\BibitemShut {NoStop}%
\bibitem [{\citenamefont {Yang}\ \emph {et~al.}(2015)\citenamefont {Yang},
  \citenamefont {Peng}, \citenamefont {Davidson},\ and\ \citenamefont
  {Yang}}]{YangPengDavidson2015ppRPAes}%
  \BibitemOpen
  \bibfield  {author} {\bibinfo {author} {\bibfnamefont {Y.}~\bibnamefont
  {Yang}}, \bibinfo {author} {\bibfnamefont {D.}~\bibnamefont {Peng}}, \bibinfo
  {author} {\bibfnamefont {E.~R.}\ \bibnamefont {Davidson}},\ and\ \bibinfo
  {author} {\bibfnamefont {W.}~\bibnamefont {Yang}},\ }\bibfield  {title}
  {\bibinfo {title} {Singlet-triplet energy gaps for diradicals from particle
  particle random phase approximation},\ }\href@noop {} {\bibfield  {journal}
  {\bibinfo  {journal} {Journal of Physical Chemistry A}\ }\textbf {\bibinfo
  {volume} {119}},\ \bibinfo {pages} {4923} (\bibinfo {year}
  {2015})}\BibitemShut {NoStop}%
\bibitem [{\citenamefont {Yang}\ \emph {et~al.}(2016)\citenamefont {Yang},
  \citenamefont {Davidson},\ and\ \citenamefont
  {Yang}}]{YangDavidsonYang2016ppRPAes}%
  \BibitemOpen
  \bibfield  {author} {\bibinfo {author} {\bibfnamefont {Y.}~\bibnamefont
  {Yang}}, \bibinfo {author} {\bibfnamefont {E.~R.}\ \bibnamefont {Davidson}},\
  and\ \bibinfo {author} {\bibfnamefont {W.}~\bibnamefont {Yang}},\ }\bibfield
  {title} {\bibinfo {title} {Nature of ground and electronic excited states of
  higher acenes},\ }\href {https://doi.org/10.1073/pnas.1606021113} {\bibfield
  {journal} {\bibinfo  {journal} {Proceedings of the National Academy of
  Sciences}\ }\textbf {\bibinfo {volume} {113}},\ \bibinfo {pages} {E5098}
  (\bibinfo {year} {2016})},\ \bibinfo {note} {doi:
  10.1073/pnas.1606021113}\BibitemShut {NoStop}%
\bibitem [{\citenamefont {Mei}\ and\ \citenamefont
  {Yang}(2019)}]{MeiYang2019qedft}%
  \BibitemOpen
  \bibfield  {author} {\bibinfo {author} {\bibfnamefont {Y.}~\bibnamefont
  {Mei}}\ and\ \bibinfo {author} {\bibfnamefont {W.}~\bibnamefont {Yang}},\
  }\bibfield  {title} {\bibinfo {title} {Excited-state potential energy
  surfaces, conical intersections, and analytical gradients from ground-state
  density functional theory},\ }\href
  {https://doi.org/10.1021/acs.jpclett.9b00712} {\bibfield  {journal} {\bibinfo
   {journal} {The Journal of Physical Chemistry Letters}\ }\textbf {\bibinfo
  {volume} {10}},\ \bibinfo {pages} {2538} (\bibinfo {year} {2019})},\ \bibinfo
  {note} {doi: 10.1021/acs.jpclett.9b00712}\BibitemShut {NoStop}%
\bibitem [{\citenamefont {van Leeuwen}(1999)}]{vanLeeuwen1999tddftrigorous}%
  \BibitemOpen
  \bibfield  {author} {\bibinfo {author} {\bibfnamefont {R.}~\bibnamefont {van
  Leeuwen}},\ }\bibfield  {title} {\bibinfo {title} {Mapping from densities to
  potentials in time-dependent density-functional theory},\ }\href
  {https://link.aps.org/doi/10.1103/PhysRevLett.82.3863} {\bibfield  {journal}
  {\bibinfo  {journal} {Physical Review Letters}\ }\textbf {\bibinfo {volume}
  {82}},\ \bibinfo {pages} {3863} (\bibinfo {year} {1999})}\BibitemShut
  {NoStop}%
\bibitem [{\citenamefont {Van~Leeuwen}(2001)}]{vanLeeuwen2001tddftrigorous}%
  \BibitemOpen
  \bibfield  {author} {\bibinfo {author} {\bibfnamefont {R.}~\bibnamefont
  {Van~Leeuwen}},\ }\bibfield  {title} {\bibinfo {title} {Key concepts in
  time-dependent density-functional theory},\ }\href@noop {} {\bibfield
  {journal} {\bibinfo  {journal} {International Journal of Modern Physics B}\
  }\textbf {\bibinfo {volume} {15}},\ \bibinfo {pages} {1969} (\bibinfo {year}
  {2001})}\BibitemShut {NoStop}%
\bibitem [{\citenamefont {Ruggenthaler}\ \emph {et~al.}(2015)\citenamefont
  {Ruggenthaler}, \citenamefont {Penz},\ and\ \citenamefont
  {Leeuwen}}]{RuggenthalerVanleeuwen2015tddftrigorous}%
  \BibitemOpen
  \bibfield  {author} {\bibinfo {author} {\bibfnamefont {M.}~\bibnamefont
  {Ruggenthaler}}, \bibinfo {author} {\bibfnamefont {M.}~\bibnamefont {Penz}},\
  and\ \bibinfo {author} {\bibfnamefont {R.~v.}\ \bibnamefont {Leeuwen}},\
  }\bibfield  {title} {\bibinfo {title} {Existence, uniqueness, and
  construction of the density-potential mapping in time-dependent
  density-functional theory},\ }\href
  {https://doi.org/10.1088/0953-8984/27/20/203202} {\bibfield  {journal}
  {\bibinfo  {journal} {Journal of Physics: Condensed Matter}\ }\textbf
  {\bibinfo {volume} {27}},\ \bibinfo {pages} {203202} (\bibinfo {year}
  {2015})},\ \bibinfo {note} {0953-8984}\BibitemShut {NoStop}%
\bibitem [{\citenamefont {Slater}(1974)}]{Slater1972AdvQuantumChem}%
  \BibitemOpen
  \bibfield  {author} {\bibinfo {author} {\bibfnamefont {J.~C.}\ \bibnamefont
  {Slater}},\ }\href@noop {} {\bibinfo {title} {The self-consistent field for
  molecules and solids: Quantum theory of molecules and solids, vol. 4}}
  (\bibinfo {year} {1974})\BibitemShut {NoStop}%
\bibitem [{\citenamefont {Slater}\ and\ \citenamefont
  {Wood}(1970)}]{SlaterWood1970SlaterTM}%
  \BibitemOpen
  \bibfield  {author} {\bibinfo {author} {\bibfnamefont {J.~C.}\ \bibnamefont
  {Slater}}\ and\ \bibinfo {author} {\bibfnamefont {J.~H.}\ \bibnamefont
  {Wood}},\ }\bibfield  {title} {\bibinfo {title} {Statistical exchange and the
  total energy of a crystal},\ }\href
  {https://doi.org/https://doi.org/10.1002/qua.560050703} {\bibfield  {journal}
  {\bibinfo  {journal} {International Journal of Quantum Chemistry}\ }\textbf
  {\bibinfo {volume} {5}},\ \bibinfo {pages} {3} (\bibinfo {year}
  {1970})}\BibitemShut {NoStop}%
\bibitem [{\citenamefont {Ziegler}\ \emph {et~al.}(1977)\citenamefont
  {Ziegler}, \citenamefont {Rauk},\ and\ \citenamefont
  {Baerends}}]{ZieglerRaukBaerends1977deltaSCF}%
  \BibitemOpen
  \bibfield  {author} {\bibinfo {author} {\bibfnamefont {T.}~\bibnamefont
  {Ziegler}}, \bibinfo {author} {\bibfnamefont {A.}~\bibnamefont {Rauk}},\ and\
  \bibinfo {author} {\bibfnamefont {E.~J.}\ \bibnamefont {Baerends}},\
  }\bibfield  {title} {\bibinfo {title} {On the calculation of multiplet
  energies by the hartree-fock-slater method},\ }\href
  {https://doi.org/10.1007/BF00551551} {\bibfield  {journal} {\bibinfo
  {journal} {Theoretica chimica acta}\ }\textbf {\bibinfo {volume} {43}},\
  \bibinfo {pages} {261} (\bibinfo {year} {1977})}\BibitemShut {NoStop}%
\bibitem [{\citenamefont {Jones}\ and\ \citenamefont
  {Gunnarsson}(1989)}]{JonesGunnarsson1989dftreview}%
  \BibitemOpen
  \bibfield  {author} {\bibinfo {author} {\bibfnamefont {R.~O.}\ \bibnamefont
  {Jones}}\ and\ \bibinfo {author} {\bibfnamefont {O.}~\bibnamefont
  {Gunnarsson}},\ }\bibfield  {title} {\bibinfo {title} {The density functional
  formalism, its applications and prospects},\ }\href
  {https://doi.org/10.1103/RevModPhys.61.689} {\bibfield  {journal} {\bibinfo
  {journal} {Reviews of Modern Physics}\ }\textbf {\bibinfo {volume} {61}},\
  \bibinfo {pages} {689} (\bibinfo {year} {1989})},\ \bibinfo {note}
  {rMP}\BibitemShut {NoStop}%
\bibitem [{\citenamefont {Triguero}\ \emph {et~al.}(1999)\citenamefont
  {Triguero}, \citenamefont {Plashkevych}, \citenamefont {Pettersson},\ and\
  \citenamefont {Ågren}}]{TrigueroPetterssonAgren1999deltaSCF}%
  \BibitemOpen
  \bibfield  {author} {\bibinfo {author} {\bibfnamefont {L.}~\bibnamefont
  {Triguero}}, \bibinfo {author} {\bibfnamefont {O.}~\bibnamefont
  {Plashkevych}}, \bibinfo {author} {\bibfnamefont {L.~G.~M.}\ \bibnamefont
  {Pettersson}},\ and\ \bibinfo {author} {\bibfnamefont {H.}~\bibnamefont
  {Ågren}},\ }\bibfield  {title} {\bibinfo {title} {Separate state vs.
  transition state kohn-sham calculations of x-ray photoelectron binding
  energies and chemical shifts},\ }\href
  {https://doi.org/https://doi.org/10.1016/S0368-2048(99)00008-0} {\bibfield
  {journal} {\bibinfo  {journal} {Journal of Electron Spectroscopy and Related
  Phenomena}\ }\textbf {\bibinfo {volume} {104}},\ \bibinfo {pages} {195}
  (\bibinfo {year} {1999})}\BibitemShut {NoStop}%
\bibitem [{\citenamefont {Cheng}\ \emph {et~al.}(2008)\citenamefont {Cheng},
  \citenamefont {Wu},\ and\ \citenamefont
  {Van~Voorhis}}]{ChengWuVanVoorhis2008DeltaSCF}%
  \BibitemOpen
  \bibfield  {author} {\bibinfo {author} {\bibfnamefont {C.-L.}\ \bibnamefont
  {Cheng}}, \bibinfo {author} {\bibfnamefont {Q.}~\bibnamefont {Wu}},\ and\
  \bibinfo {author} {\bibfnamefont {T.}~\bibnamefont {Van~Voorhis}},\
  }\bibfield  {title} {\bibinfo {title} {Rydberg energies using excited state
  density functional theory},\ }\href {https://doi.org/10.1063/1.2977989}
  {\bibfield  {journal} {\bibinfo  {journal} {The Journal of Chemical Physics}\
  }\textbf {\bibinfo {volume} {129}},\ \bibinfo {pages} {124112} (\bibinfo
  {year} {2008})}\BibitemShut {NoStop}%
\bibitem [{\citenamefont {Gilbert}\ \emph {et~al.}(2008)\citenamefont
  {Gilbert}, \citenamefont {Besley},\ and\ \citenamefont
  {Gill}}]{GilbertBesleyGill2008MOM}%
  \BibitemOpen
  \bibfield  {author} {\bibinfo {author} {\bibfnamefont {A.~T.~B.}\
  \bibnamefont {Gilbert}}, \bibinfo {author} {\bibfnamefont {N.~A.}\
  \bibnamefont {Besley}},\ and\ \bibinfo {author} {\bibfnamefont {P.~M.~W.}\
  \bibnamefont {Gill}},\ }\bibfield  {title} {\bibinfo {title} {Self-consistent
  field calculations of excited states using the maximum overlap method
  (mom)},\ }\href {https://doi.org/10.1021/jp801738f} {\bibfield  {journal}
  {\bibinfo  {journal} {The Journal of Physical Chemistry A}\ }\textbf
  {\bibinfo {volume} {112}},\ \bibinfo {pages} {13164} (\bibinfo {year}
  {2008})},\ \bibinfo {note} {doi: 10.1021/jp801738f}\BibitemShut {NoStop}%
\bibitem [{\citenamefont {Seidu}\ \emph {et~al.}(2015)\citenamefont {Seidu},
  \citenamefont {Krykunov},\ and\ \citenamefont
  {Ziegler}}]{SeiduZiegler2015constrictedDFT}%
  \BibitemOpen
  \bibfield  {author} {\bibinfo {author} {\bibfnamefont {I.}~\bibnamefont
  {Seidu}}, \bibinfo {author} {\bibfnamefont {M.}~\bibnamefont {Krykunov}},\
  and\ \bibinfo {author} {\bibfnamefont {T.}~\bibnamefont {Ziegler}},\
  }\bibfield  {title} {\bibinfo {title} {Applications of time-dependent and
  time-independent density functional theory to rydberg transitions},\ }\href
  {https://doi.org/10.1021/jp5082802} {\bibfield  {journal} {\bibinfo
  {journal} {The Journal of Physical Chemistry A}\ }\textbf {\bibinfo {volume}
  {119}},\ \bibinfo {pages} {5107} (\bibinfo {year} {2015})},\ \bibinfo {note}
  {doi: 10.1021/jp5082802}\BibitemShut {NoStop}%
\bibitem [{\citenamefont {Liu}\ \emph {et~al.}(2017)\citenamefont {Liu},
  \citenamefont {Zhang}, \citenamefont {Bao},\ and\ \citenamefont
  {Yi}}]{LiuBaoYi2017deltaSCF}%
  \BibitemOpen
  \bibfield  {author} {\bibinfo {author} {\bibfnamefont {J.}~\bibnamefont
  {Liu}}, \bibinfo {author} {\bibfnamefont {Y.}~\bibnamefont {Zhang}}, \bibinfo
  {author} {\bibfnamefont {P.}~\bibnamefont {Bao}},\ and\ \bibinfo {author}
  {\bibfnamefont {Y.}~\bibnamefont {Yi}},\ }\bibfield  {title} {\bibinfo
  {title} {Evaluating electronic couplings for excited state charge transfer
  based on maximum occupation method delta scf quasi-adiabatic states},\ }\href
  {https://doi.org/10.1021/acs.jctc.6b01161} {\bibfield  {journal} {\bibinfo
  {journal} {Journal of Chemical Theory and Computation}\ }\textbf {\bibinfo
  {volume} {13}},\ \bibinfo {pages} {843} (\bibinfo {year} {2017})},\ \bibinfo
  {note} {doi: 10.1021/acs.jctc.6b01161}\BibitemShut {NoStop}%
\bibitem [{\citenamefont {Barca}\ \emph
  {et~al.}(2018{\natexlab{a}})\citenamefont {Barca}, \citenamefont {Gilbert},\
  and\ \citenamefont {Gill}}]{BarcaGilbertGill2018deltaSCFes}%
  \BibitemOpen
  \bibfield  {author} {\bibinfo {author} {\bibfnamefont {G.~M.~J.}\
  \bibnamefont {Barca}}, \bibinfo {author} {\bibfnamefont {A.~T.~B.}\
  \bibnamefont {Gilbert}},\ and\ \bibinfo {author} {\bibfnamefont {P.~M.~W.}\
  \bibnamefont {Gill}},\ }\bibfield  {title} {\bibinfo {title} {Excitation
  number: Characterizing multiply excited states},\ }\href
  {https://doi.org/10.1021/acs.jctc.7b00963} {\bibfield  {journal} {\bibinfo
  {journal} {Journal of Chemical Theory and Computation}\ }\textbf {\bibinfo
  {volume} {14}},\ \bibinfo {pages} {9} (\bibinfo {year}
  {2018}{\natexlab{a}})},\ \bibinfo {note} {doi:
  10.1021/acs.jctc.7b00963}\BibitemShut {NoStop}%
\bibitem [{\citenamefont {Barca}\ \emph
  {et~al.}(2018{\natexlab{b}})\citenamefont {Barca}, \citenamefont {Gilbert},\
  and\ \citenamefont {Gill}}]{BarcaGilbertGill2018iMOM}%
  \BibitemOpen
  \bibfield  {author} {\bibinfo {author} {\bibfnamefont {G.~M.~J.}\
  \bibnamefont {Barca}}, \bibinfo {author} {\bibfnamefont {A.~T.~B.}\
  \bibnamefont {Gilbert}},\ and\ \bibinfo {author} {\bibfnamefont {P.~M.~W.}\
  \bibnamefont {Gill}},\ }\bibfield  {title} {\bibinfo {title} {Simple models
  for difficult electronic excitations},\ }\href
  {https://doi.org/10.1021/acs.jctc.7b00994} {\bibfield  {journal} {\bibinfo
  {journal} {Journal of Chemical Theory and Computation}\ }\textbf {\bibinfo
  {volume} {14}},\ \bibinfo {pages} {1501} (\bibinfo {year}
  {2018}{\natexlab{b}})},\ \bibinfo {note} {doi:
  10.1021/acs.jctc.7b00994}\BibitemShut {NoStop}%
\bibitem [{\citenamefont {Hait}\ and\ \citenamefont
  {Head-Gordon}(2020)}]{HaitHeadgordon2020deltaSCF}%
  \BibitemOpen
  \bibfield  {author} {\bibinfo {author} {\bibfnamefont {D.}~\bibnamefont
  {Hait}}\ and\ \bibinfo {author} {\bibfnamefont {M.}~\bibnamefont
  {Head-Gordon}},\ }\bibfield  {title} {\bibinfo {title} {Highly accurate
  prediction of core spectra of molecules at density functional theory cost:
  Attaining sub-electronvolt error from a restricted open-shell kohn–sham
  approach},\ }\href {https://doi.org/10.1021/acs.jpclett.9b03661} {\bibfield
  {journal} {\bibinfo  {journal} {The Journal of Physical Chemistry Letters}\
  }\textbf {\bibinfo {volume} {11}},\ \bibinfo {pages} {775} (\bibinfo {year}
  {2020})},\ \bibinfo {note} {doi: 10.1021/acs.jpclett.9b03661}\BibitemShut
  {NoStop}%
\bibitem [{\citenamefont {Kumar}\ and\ \citenamefont
  {Luber}(2022)}]{KumarLuber2022deltaSCF}%
  \BibitemOpen
  \bibfield  {author} {\bibinfo {author} {\bibfnamefont {C.}~\bibnamefont
  {Kumar}}\ and\ \bibinfo {author} {\bibfnamefont {S.}~\bibnamefont {Luber}},\
  }\bibfield  {title} {\bibinfo {title} {Robust Delta scf calculations with direct
  energy functional minimization methods and step for molecules and
  materials},\ }\href {https://doi.org/10.1063/5.0075927} {\bibfield  {journal}
  {\bibinfo  {journal} {The Journal of Chemical Physics}\ }\textbf {\bibinfo
  {volume} {156}},\ \bibinfo {pages} {154104} (\bibinfo {year}
  {2022})}\BibitemShut {NoStop}%
\bibitem [{\citenamefont {Kunze}\ \emph {et~al.}(2021)\citenamefont {Kunze},
  \citenamefont {Hansen}, \citenamefont {Grimme},\ and\ \citenamefont
  {Mewes}}]{GrimmeMewes2021deltaSCF}%
  \BibitemOpen
  \bibfield  {author} {\bibinfo {author} {\bibfnamefont {L.}~\bibnamefont
  {Kunze}}, \bibinfo {author} {\bibfnamefont {A.}~\bibnamefont {Hansen}},
  \bibinfo {author} {\bibfnamefont {S.}~\bibnamefont {Grimme}},\ and\ \bibinfo
  {author} {\bibfnamefont {J.-M.}\ \bibnamefont {Mewes}},\ }\bibfield  {title}
  {\bibinfo {title} {Pcm-roks for the description of charge-transfer states in
  solution: Singlet–triplet gaps with chemical accuracy from open-shell
  kohn–sham reaction-field calculations},\ }\href
  {https://doi.org/10.1021/acs.jpclett.1c02299} {\bibfield  {journal} {\bibinfo
   {journal} {The Journal of Physical Chemistry Letters}\ }\textbf {\bibinfo
  {volume} {12}},\ \bibinfo {pages} {8470} (\bibinfo {year} {2021})},\ \bibinfo
  {note} {doi: 10.1021/acs.jpclett.1c02299}\BibitemShut {NoStop}%
\bibitem [{\citenamefont {Perdew}\ and\ \citenamefont
  {Levy}(1985)}]{PerdewLevy1985esDFT}%
  \BibitemOpen
  \bibfield  {author} {\bibinfo {author} {\bibfnamefont {J.~P.}\ \bibnamefont
  {Perdew}}\ and\ \bibinfo {author} {\bibfnamefont {M.}~\bibnamefont {Levy}},\
  }\bibfield  {title} {\bibinfo {title} {Extrema of the density functional for
  the energy: excited states from the ground-state theory},\ }\href@noop {}
  {\bibfield  {journal} {\bibinfo  {journal} {Physical Review B}\ }\textbf
  {\bibinfo {volume} {31}},\ \bibinfo {pages} {6264} (\bibinfo {year}
  {1985})}\BibitemShut {NoStop}%
\bibitem [{\citenamefont {Vandaele}\ \emph {et~al.}(2022)\citenamefont
  {Vandaele}, \citenamefont {Mališ},\ and\ \citenamefont
  {Luber}}]{VandaeleLuber2002deltaSCFreview}%
  \BibitemOpen
  \bibfield  {author} {\bibinfo {author} {\bibfnamefont {E.}~\bibnamefont
  {Vandaele}}, \bibinfo {author} {\bibfnamefont {M.}~\bibnamefont {Mališ}},\
  and\ \bibinfo {author} {\bibfnamefont {S.}~\bibnamefont {Luber}},\ }\bibfield
   {title} {\bibinfo {title} {The $\Delta$scf method for non-adiabatic dynamics of
  systems in the liquid phase},\ }\href {https://doi.org/10.1063/5.0083340}
  {\bibfield  {journal} {\bibinfo  {journal} {The Journal of Chemical Physics}\
  }\textbf {\bibinfo {volume} {156}},\ \bibinfo {pages} {130901} (\bibinfo
  {year} {2022})}\BibitemShut {NoStop}%
\bibitem [{\citenamefont {Herbert}(2023)}]{Herbert2023esDFTchapter}%
  \BibitemOpen
  \bibfield  {author} {\bibinfo {author} {\bibfnamefont {J.~M.}\ \bibnamefont
  {Herbert}},\ }\bibinfo {title} {Chapter 3 - density-functional theory for
  electronic excited states},\ in\ \href
  {https://doi.org/https://doi.org/10.1016/B978-0-323-91738-4.00005-1} {\emph
  {\bibinfo {booktitle} {Theoretical and Computational Photochemistry}}},\
  \bibinfo {editor} {edited by\ \bibinfo {editor} {\bibfnamefont
  {C.}~\bibnamefont {García-Iriepa}}\ and\ \bibinfo {editor} {\bibfnamefont
  {M.}~\bibnamefont {Marazzi}}}\ (\bibinfo  {publisher} {Elsevier},\ \bibinfo
  {year} {2023})\ pp.\ \bibinfo {pages} {69--118}\BibitemShut {NoStop}%
\bibitem [{\citenamefont {Yang}\ \emph {et~al.}(2004)\citenamefont {Yang},
  \citenamefont {Ayers},\ and\ \citenamefont
  {Wu}}]{YangAyersWu2004potentialFT}%
  \BibitemOpen
  \bibfield  {author} {\bibinfo {author} {\bibfnamefont {W.~T.}\ \bibnamefont
  {Yang}}, \bibinfo {author} {\bibfnamefont {P.~W.}\ \bibnamefont {Ayers}},\
  and\ \bibinfo {author} {\bibfnamefont {Q.}~\bibnamefont {Wu}},\ }\bibfield
  {title} {\bibinfo {title} {Potential functionals: Dual to density functionals
  and solution to the upsilon-representability problem},\ }\href@noop {}
  {\bibfield  {journal} {\bibinfo  {journal} {Physical Review Letters}\
  }\textbf {\bibinfo {volume} {92}},\ \bibinfo {pages} {146404} (\bibinfo
  {year} {2004})}\BibitemShut {NoStop}%
\bibitem [{\citenamefont {Langreth}\ and\ \citenamefont
  {Perdew}(1977)}]{LangrethPerdew1977adiabaticC}%
  \BibitemOpen
  \bibfield  {author} {\bibinfo {author} {\bibfnamefont {D.~C.}\ \bibnamefont
  {Langreth}}\ and\ \bibinfo {author} {\bibfnamefont {J.~P.}\ \bibnamefont
  {Perdew}},\ }\bibfield  {title} {\bibinfo {title} {Exchange-correlation
  energy ofa metallic surface: Wave-vector analysis},\ }\href@noop {}
  {\bibfield  {journal} {\bibinfo  {journal} {Physical Review B}\ }\textbf
  {\bibinfo {volume} {15}},\ \bibinfo {pages} {2884} (\bibinfo {year}
  {1977})}\BibitemShut {NoStop}%
\bibitem [{\citenamefont
  {Harris}(1984)}]{harrisAdiabaticconnectionApproachKohnSham1984a}%
  \BibitemOpen
  \bibfield  {author} {\bibinfo {author} {\bibfnamefont {J.}~\bibnamefont
  {Harris}},\ }\bibfield  {title} {\bibinfo {title} {Adiabatic-connection
  approach to {{Kohn-Sham}} theory},\ }\href@noop {} {\bibfield  {journal}
  {\bibinfo  {journal} {Physical Review A}\ }\textbf {\bibinfo {volume} {29}},\
  \bibinfo {pages} {1648} (\bibinfo {year} {1984})}\BibitemShut {NoStop}%
\bibitem [{\citenamefont
  {Yang}(1998)}]{yangGeneralizedAdiabaticConnection1998a}%
  \BibitemOpen
  \bibfield  {author} {\bibinfo {author} {\bibfnamefont {W.}~\bibnamefont
  {Yang}},\ }\bibfield  {title} {\bibinfo {title} {Generalized adiabatic
  connection in density functional theory},\ }\href@noop {} {\bibfield
  {journal} {\bibinfo  {journal} {Journal of Chemical Physics}\ }\textbf
  {\bibinfo {volume} {109}},\ \bibinfo {pages} {10107} (\bibinfo {year}
  {1998})}\BibitemShut {NoStop}%
\bibitem [{\citenamefont {Harris}\ and\ \citenamefont
  {Jones}(1974)}]{harrisSurfaceEnergyBounded1974a}%
  \BibitemOpen
  \bibfield  {author} {\bibinfo {author} {\bibfnamefont {J.}~\bibnamefont
  {Harris}}\ and\ \bibinfo {author} {\bibfnamefont {R.~O.}\ \bibnamefont
  {Jones}},\ }\bibfield  {title} {\bibinfo {title} {The surface energy of a
  bounded electron gas},\ }\href@noop {} {\bibfield  {journal} {\bibinfo
  {journal} {J.Phys.F}\ }\textbf {\bibinfo {volume} {4}},\ \bibinfo {pages}
  {1170} (\bibinfo {year} {1974})}\BibitemShut {NoStop}%
\bibitem [{\citenamefont {Sharp}\ and\ \citenamefont
  {Horton}(1953)}]{sharpVariationalApproachUnipotential1953}%
  \BibitemOpen
  \bibfield  {author} {\bibinfo {author} {\bibfnamefont {R.~T.}\ \bibnamefont
  {Sharp}}\ and\ \bibinfo {author} {\bibfnamefont {G.~K.}\ \bibnamefont
  {Horton}},\ }\bibfield  {title} {\bibinfo {title} {A variational approach to
  the unipotential many-electron problem},\ }\href@noop {} {\bibfield
  {journal} {\bibinfo  {journal} {Physical Review}\ }\textbf {\bibinfo {volume}
  {90}},\ \bibinfo {pages} {317} (\bibinfo {year} {1953})}\BibitemShut
  {NoStop}%
\bibitem [{\citenamefont {Talman}\ and\ \citenamefont
  {Shadwick}(1976)}]{talmanOptimizedEffectiveAtomic1976}%
  \BibitemOpen
  \bibfield  {author} {\bibinfo {author} {\bibfnamefont {J.~D.}\ \bibnamefont
  {Talman}}\ and\ \bibinfo {author} {\bibfnamefont {W.~F.}\ \bibnamefont
  {Shadwick}},\ }\bibfield  {title} {\bibinfo {title} {Optimized effective
  atomic central potential},\ }\href@noop {} {\bibfield  {journal} {\bibinfo
  {journal} {Physical Review A}\ }\textbf {\bibinfo {volume} {14}},\ \bibinfo
  {pages} {36} (\bibinfo {year} {1976})}\BibitemShut {NoStop}%
\bibitem [{\citenamefont {Jin}\ \emph {et~al.}(2017)\citenamefont {Jin},
  \citenamefont {Zhang}, \citenamefont {Chen}, \citenamefont {Su},\ and\
  \citenamefont {Yang}}]{jinGeneralizedOptimizedEffective2017}%
  \BibitemOpen
  \bibfield  {author} {\bibinfo {author} {\bibfnamefont {Y.}~\bibnamefont
  {Jin}}, \bibinfo {author} {\bibfnamefont {D.}~\bibnamefont {Zhang}}, \bibinfo
  {author} {\bibfnamefont {Z.}~\bibnamefont {Chen}}, \bibinfo {author}
  {\bibfnamefont {N.~Q.}\ \bibnamefont {Su}},\ and\ \bibinfo {author}
  {\bibfnamefont {W.}~\bibnamefont {Yang}},\ }\bibfield  {title} {\bibinfo
  {title} {Generalized {{Optimized Effective Potential}} for {{Orbital
  Functionals}} and {{Self-Consistent Calculation}} of {{Random Phase
  Approximations}}},\ }\href {https://doi.org/10.1021/acs.jpclett.7b02165}
  {\bibfield  {journal} {\bibinfo  {journal} {The Journal of Physical Chemistry
  Letters}\ }\textbf {\bibinfo {volume} {8}},\ \bibinfo {pages} {4746}
  (\bibinfo {year} {2017})}\BibitemShut {NoStop}%
\bibitem [{\citenamefont {Parr}\ and\ \citenamefont
  {Yang}(1989)}]{parrDensityFunctionalTheoryAtoms1989}%
  \BibitemOpen
  \bibfield  {author} {\bibinfo {author} {\bibfnamefont {R.~G.}\ \bibnamefont
  {Parr}}\ and\ \bibinfo {author} {\bibfnamefont {W.}~\bibnamefont {Yang}},\
  }\href@noop {} {\emph {\bibinfo {title} {Density-{{Functional Theory}} of
  {{Atoms}} and {{Molecules}}}}}\ (\bibinfo  {publisher} {{Oxford UP}},\
  \bibinfo {address} {{New York}},\ \bibinfo {year} {1989})\BibitemShut
  {NoStop}%
\bibitem [{\citenamefont {Cohen}\ \emph {et~al.}(2008)\citenamefont {Cohen},
  \citenamefont {{Mori-Sanchez}},\ and\ \citenamefont
  {Yang}}]{cohenFractionalChargePerspective2008b}%
  \BibitemOpen
  \bibfield  {author} {\bibinfo {author} {\bibfnamefont {A.~J.}\ \bibnamefont
  {Cohen}}, \bibinfo {author} {\bibfnamefont {P.}~\bibnamefont
  {{Mori-Sanchez}}},\ and\ \bibinfo {author} {\bibfnamefont {W.~T.}\
  \bibnamefont {Yang}},\ }\bibfield  {title} {\bibinfo {title} {Fractional
  charge perspective on the band gap in density-functional theory},\ }\href
  {https://doi.org/115123 Artn 115123} {\bibfield  {journal} {\bibinfo
  {journal} {Physical Review B}\ }\textbf {\bibinfo {volume} {77}},\ \bibinfo
  {pages} {115123} (\bibinfo {year} {2008})}\BibitemShut {NoStop}%
\bibitem [{\citenamefont {Garrick}\ \emph {et~al.}(2020)\citenamefont
  {Garrick}, \citenamefont {Natan}, \citenamefont {Gould},\ and\ \citenamefont
  {Kronik}}]{garrickExactGeneralizedKohnSham2020}%
  \BibitemOpen
  \bibfield  {author} {\bibinfo {author} {\bibfnamefont {R.}~\bibnamefont
  {Garrick}}, \bibinfo {author} {\bibfnamefont {A.}~\bibnamefont {Natan}},
  \bibinfo {author} {\bibfnamefont {T.}~\bibnamefont {Gould}},\ and\ \bibinfo
  {author} {\bibfnamefont {L.}~\bibnamefont {Kronik}},\ }\bibfield  {title}
  {\bibinfo {title} {Exact {{Generalized Kohn-Sham Theory}} for {{Hybrid
  Functionals}}},\ }\href {https://doi.org/10.1103/PhysRevX.10.021040}
  {\bibfield  {journal} {\bibinfo  {journal} {Physical Review X}\ }\textbf
  {\bibinfo {volume} {10}},\ \bibinfo {pages} {021040} (\bibinfo {year}
  {2020})}\BibitemShut {NoStop}%
\bibitem [{\citenamefont {Jin}\ \emph {et~al.}(2020)\citenamefont {Jin},
  \citenamefont {Su}, \citenamefont {Chen},\ and\ \citenamefont
  {Yang}}]{jinIntroductoryLectureWhen2020}%
  \BibitemOpen
  \bibfield  {author} {\bibinfo {author} {\bibfnamefont {Y.}~\bibnamefont
  {Jin}}, \bibinfo {author} {\bibfnamefont {N.~Q.}\ \bibnamefont {Su}},
  \bibinfo {author} {\bibfnamefont {Z.}~\bibnamefont {Chen}},\ and\ \bibinfo
  {author} {\bibfnamefont {W.}~\bibnamefont {Yang}},\ }\bibfield  {title}
  {\bibinfo {title} {Introductory lecture: When the density of the
  noninteracting reference system is not the density of the physical system in
  density functional theory},\ }\href {https://doi.org/10.1039/D0FD00102C}
  {\bibfield  {journal} {\bibinfo  {journal} {Faraday Discussions}\ }\textbf
  {\bibinfo {volume} {224}},\ \bibinfo {pages} {9} (\bibinfo {year}
  {2020})}\BibitemShut {NoStop}%
\bibitem [{\citenamefont {Kowalczyk}\ \emph {et~al.}(2011)\citenamefont
  {Kowalczyk}, \citenamefont {Yost},\ and\ \citenamefont
  {Voorhis}}]{KowalczykanVoorhis2011deltaSCF}%
  \BibitemOpen
  \bibfield  {author} {\bibinfo {author} {\bibfnamefont {T.}~\bibnamefont
  {Kowalczyk}}, \bibinfo {author} {\bibfnamefont {S.~R.}\ \bibnamefont
  {Yost}},\ and\ \bibinfo {author} {\bibfnamefont {T.~V.}\ \bibnamefont
  {Voorhis}},\ }\bibfield  {title} {\bibinfo {title} {Assessment of the $\Delta$scf
  density functional theory approach for electronic excitations in organic
  dyes},\ }\href {https://doi.org/10.1063/1.3530801} {\bibfield  {journal}
  {\bibinfo  {journal} {The Journal of Chemical Physics}\ }\textbf {\bibinfo
  {volume} {134}},\ \bibinfo {pages} {054128} (\bibinfo {year}
  {2011})}\BibitemShut {NoStop}%
\end{thebibliography}
\providecommand{\noopsort}[1]{}\providecommand{\singleletter}[1]{#1}%

\end{document}